\documentclass[11pt]{article}

\usepackage[margin=1in]{geometry}
\usepackage{amsmath,amssymb,amsthm,mathtools}
\usepackage{microtype}
\usepackage[colorlinks=true,linkcolor=blue,citecolor=blue,urlcolor=blue]{hyperref}
\usepackage{cleveref}
\usepackage{enumitem}
\usepackage{booktabs,tabularx,array}
\usepackage{thm-restate}

\usepackage{tikz}
\usetikzlibrary{patterns}

\usepackage[skins]{tcolorbox}
\tcbuselibrary{breakable}

\newtcolorbox[
  auto counter, number within=section
]{algorithmbox}[3][]{%
  enhanced,
  colback=white,colframe=black,coltitle=black,
  sharp corners,boxrule=0.4pt,
  fonttitle=\itshape,
  attach boxed title to top left={yshift=-0.3\baselineskip-0.4pt,xshift=2mm},
  boxed title style={tile,size=minimal,left=0.5mm,right=0.5mm,
                     colback=white,before upper=\strut},
  float*=htb,                         
  title   ={Protocol~\thetcbcounter: #2},                       
  label   ={#3},
  #1                                   
}

\newtheorem{theorem}{Theorem}
\newtheorem{lemma}[theorem]{Lemma}
\newtheorem{proposition}[theorem]{Proposition}
\newtheorem{corollary}[theorem]{Corollary}

\newtheorem*{informaltheorem}{Theorem}
\theoremstyle{definition}
\newtheorem{definition}[theorem]{Definition}
\theoremstyle{remark}

\newcommand{\QMAplus}{\mathsf{QMA}^{+}}
\newcommand{\QMAreal}{\mathsf{QMA}^{\mathbb R}}
\newcommand{\QMAtwo}{\mathsf{QMA}(2)}
\newcommand{\NEXP}{\mathsf{NEXP}}
\newcommand{\Sym}{\operatorname{Sym}}
\newcommand{\D}{\mathsf D}
\newcommand{\Tr}{\operatorname{Tr}}
\newcommand{\E}{\mathbb E}
\newcommand{\Prb}{\mathbb P}
\newcommand{\one}{\mathbf 1}
\newcommand{\ket}[1]{\lvert #1\rangle}
\newcommand{\bra}[1]{\langle #1\rvert}
\newcommand{\braket}[2]{\langle #1\mid #2\rangle}
\newcommand{\proj}[1]{\lvert #1\rangle\!\langle #1\rvert}

\title{Near-Optimal Gap Amplification for Nonnegative Unentangled Quantum Proofs}
\author{
Masayuki Miyamoto
\\University of Tsukuba, Japan
 }
\date{}

\begin{document}
\maketitle

\begin{abstract}
We study gap amplification of the class $\mathsf{QMA}^{+}(2)$ characterized by unentangled quantum proofs whose amplitudes are nonnegative
in the computational basis. This class was recently introduced by Jeronimo and Wu (STOC 2023), and its behavior depends sharply on the completeness-soundness gap: although it captures the power of $\mathsf{NEXP}$ for some small constant gap, it is equal to $\mathsf{QMA}(2)$ for larger constant gap. This is in stark contrast to $\mathsf{QMA}(2)$ where strong gap amplification is known due to the product test by Harrow and Montanaro (FOCS 2010, JACM 2013).

In this paper, we prove for every completeness $c$ and soundness $s$ with $c-s=1/\mathrm{poly}(n)$,
\[
  \NEXP
  = \QMAplus(2,c,s)
  =
  \QMAplus\!\left(2,1-\frac1{\mathrm{poly}(n)},\frac14+\frac1{\mathrm{poly}(n)}\right).
\]
Our result gives a clean complexity phase transition for $\QMAplus(2)$ since we have
\[
  \QMAreal(2)
  =
  \QMAplus\!\left(2,1-\frac1{\mathrm{poly}(n)},\frac14-\frac1{\mathrm{poly}(n)}\right),
\]
where $\QMAreal(2)$ denotes $\QMAplus(2)$ with witnesses restricted to real amplitudes.
Our amplification is thus optimal in the sense that a slight improvement of our soundness would have the collapse
\[
\QMAreal(2)=\NEXP.
\]

Our proof combines symmetric-subspace projections with the relation $\QMAplus(1)=\NEXP$ of Bassirian, Fefferman, and Marwaha (ITCS 2024).
The main technical ingredient is a dimension-independent de Finetti theorem in Hilbert-Schmidt norm that applies when the number of registers under consideration grows logarithmically.  
\end{abstract}

\newpage
\tableofcontents

\section{Introduction}
Quantum entanglement has been one of the most useful resources used in information processing. In the field of quantum complexity theory, the power of entanglement has been studied for many years. One of the major discoveries to emerge from recent research is the class $\mathsf{MIP}^*$ of multi-witness interactive proofs involving quantum entanglement. This class is capable of deciding the notorious halting problem, which is undecidable even in unlimited computational time~\cite{ji2021mip}. In contrast, it is also known that quantum {\it unentanglement} can sometimes provide strong computational power to proof systems. This notion of unentanglement is characterized by the complexity class $\mathsf{QMA}(2)$~\cite{kobayashi2003quantum,aaronson2009power}, where two unentangled provers provide quantum proofs to the verifier. Settling the complexity of $\mathsf{QMA}(2)$ has been a major open problem in quantum complexity theory. Currently, we only know trivial relation: 
\[
\mathsf{QMA} \subseteq \mathsf{QMA}(2)\subseteq \mathsf{NEXP}.
\]

In this paper, we focus on gap amplification, a fundamental property of proof systems. 
In the classical complexity theory, gap amplification is one of the basic robustness properties of proof systems: for example, in classical randomized proof systems, namely $\mathsf{MA}$, an inverse polynomial completeness-soundness gap can be amplified to exponentially small error by standard repetition procedures (running the verification protocol multiple times using the same classical proof). In \(\mathsf{QMA}\), the Marriott--Watrous amplification theorem~\cite{MarriottWatrous05} achieves this without increasing the number of quantum proofs. Moreover, gap amplification is of course standard for computations without proofs, such as $\mathsf{BPP}$ and $\mathsf{BQP}$. Thus the precise choice of constant error parameters is largely immaterial for a wide range of classes.

The situation changes dramatically for unentangled quantum proofs.  The class
\(\mathsf{QMA}(2)\), introduced by Kobayashi, Matsumoto, and Yamakami~\cite{kobayashi2003quantum} and Aaronson, Beigi, Drucker, Fefferman, and Shor~\cite{aaronson2009power}, is the two-prover analogue of \(\mathsf{QMA}\) in which Arthur receives two quantum proofs promised to be unentangled.  Classically, several Merlins can be merged into one. So we have $\mathsf{MA}(2)=\mathsf{MA}$ and $\mathsf{QCMA}(2)=\mathsf{QCMA}$. 
Quantumly, the product constraint is a genuine geometric condition: the verifier optimizes not over all states but over product states, and even the robustness of \(\mathsf{QMA}(2)\) under changes of error parameters was open for years. The obstacle is simple to state. In the single prover setting, sending multiple quantum proofs and repeting the verification protocol multiple times work well, but in the
two-prover setting the same does not hold: after the verifier performs measurements to a part of multiple unentangled proofs from two provers, the remaining part of quantum proofs can be entangled. This phenomenon is known as {\it entanglement swapping}. Harrow and Montanaro~\cite{harrow2013testing} resolved this problem by introducing the product test.
Since the product test implements a separable measurement, it can avoid the problem of entanglement swapping: separable measurements can be sequentially repeated while preserving separability on the unmeasured registers. Their analysis showed both strong gap amplification for
\(\mathsf{QMA}(2)\) and the collapse
\[
        \mathsf{QMA}(k)=\mathsf{QMA}(2)
\]
for polynomially bounded \(k\).  Conceptually, the product test supplies the missing operation: it enforces multipartite product structure using only two
Merlins. 

This paper studies the corresponding amplification problem for
\(\mathsf{QMA}^{+}(k)\), the nonnegative-amplitude variant introduced by
Jeronimo and Wu~\cite{jeronimo2023power}.  In \(\mathsf{QMA}^{+}(2)\), each witness is required to have
nonnegative real amplitudes in the computational basis.  This restriction is not
a harmless simplification. In fact, Jeronimo and Wu~\cite{jeronimo2023power} proved that, for an appropriate
constant gap,
\[
        \mathsf{QMA}^{+}(2)=\mathsf{NEXP}.
\]
This gives a mild gap amplification of $\mathsf{QMA}^{+}(2)$. Since $\mathsf{QMA}^{+}(2)$ is contained in $\mathsf{NEXP}$ for any completeness and soundness with an inverse-polynomial gap, their result also implies an inverse-polynomial gap can be amplified to a certain constant gap. 
They also identified a larger-gap regime in which \(\mathsf{QMA}^{+}(2)\)
collapses back toward \(\mathsf{QMA}(2)\). 
Thus \(\mathsf{QMA}^{+}(2)\) is highly gap-sensitive: improving the gap too far would
have major consequences for the longstanding question of whether
\(\mathsf{QMA}(2)\) reaches \(\mathsf{NEXP}\).

A parallel development comes from disentanglers.  One way to compare
\(\mathsf{QMA}\) and \(\mathsf{QMA}(2)\) is to ask whether a single quantum proof
can be efficiently transformed into two approximately separable proofs.  The
disentangler conjecture asserts that such a transformation should require
exponential dimension blowup. Jeronimo and Wu's dimension-independent
disentangler work~\cite{jeronimo2024dimension} gives a different type of result: using bipartite unentanglement as input, they construct channels that output states close to
multipartite product mixtures, and use this to reprove the result of $\mathsf{QMA}(k)=\mathsf{QMA}(2)$. 
Using this disentangler, they also constructed a $\mathsf{QMA}^{+}(3)$ protocol for $\mathsf{NEXP}$, with completeness $1-1/\exp(n)$ and soundness $1/2+1/\mathrm{poly}(n)$. Here, soundness of this protocol does hold even for a relaxed version of unenetangled proofs, namely, one has nonnegative amplitudes and the other two have arbitrary amplitudes. This is particularly important because this type of protocol can be simulated by a $\mathsf{QMA}^{\mathbb{R}}(3)$ protocol at the cost of doubling the soundness parameter. As a result, if the soundness of their protocol can be slightly improved to $1/2-1/\mathrm{poly}(n)$, then it can be converted to a $\mathsf{QMA}^{\mathbb{R}}(3)$ protocol for $\mathsf{NEXP}$. This demonstrates a limitation of their gap amplification, unless $\mathsf{QMA}^{\mathbb{R}}(3) = \mathsf{NEXP}$.

Recent work by Bassirian, Fefferman and Marwaha~\cite{bassirian2024quantum} focused on $\mathsf{QMA}^{+}$, the nonnegative variant of $\mathsf{QMA}$. They used similar techniques from~\cite{jeronimo2023power}, and showed that $\mathsf{QMA}^{+} = \mathsf{NEXP}$ for a certain constant gap. Since any $\mathsf{QMA}^{+}$ protocol with completeness $c$ and soundness $s$ can be simulated by a $\mathsf{QMA}$ protocol with completeness $c$ and soundness $4s$, it demonstrates that $\mathsf{QMA}^{+}$ does not have strong gap amplification unless $\mathsf{QMA} = \mathsf{NEXP}$ which is highly unbelievable (as it implies e.g., $\mathsf{PP}=\mathsf{NEXP}$). It thus demonstrates that we need to exploit some structure arise from unentanglement in order to obtain gap amplification for $\mathsf{QMA}^{+}(k)$.

\subsection{Our results}
Given these evidences of the subtlety of gap amplification, a natural question to ask is
\begin{center}
    { \it
How much can we amplify the gap between completeness $c$ and soundness $s$, while maintaining the relation $\mathsf{QMA}^+(k,c,s)=\mathsf{NEXP}$?
}
\end{center}
Our result answers this question {\it almost optimally}.

Our main result is the following characterization of \(\NEXP\) with
near-perfect completeness and soundness approaching \(1/4\).

\begin{theorem}[Main theorem]
\label{thm:main}
For every prescribed polynomial \(p\) satisfying \(p(n)\ge4\) for
all sufficiently large \(n\) and \(p(n)\to\infty\),
\[
  \NEXP
  =
  \QMAplus\!\left(2,1-\frac1{p(n)},\frac14+\frac1{p(n)}\right).
\]
\end{theorem}

Since \(\QMAplus(k,c,s)\subseteq\NEXP\) for polynomially many proofs, which is the trivial upper bound for this class, the theorem gives the following.

\par\medskip
\begin{corollary}[Gap amplification of $\QMAplus(k)$]
\label{cor:general-amplification}
Let \(k=k(n)\) be polynomially bounded, and suppose that
\(c(n)-s(n)\ge 1/\operatorname{poly}(n)\).  Then, for every \(p\) as
in Theorem~\ref{thm:main},
\[
  \QMAplus(k,c,s)
  \subseteq
  \QMAplus\!\left(2,1-\frac1{p(n)},\frac14+\frac1{p(n)}\right).
\]
\end{corollary}

As mentioned earlier, previous work provided the gap can be $1-1/\exp(n)$ vs $7/9 + 1/\mathrm{poly}(n)$ for any $k\geq 2$ based on the product test, and $1-1/\exp(n)$ vs $1/2 + 1/\mathrm{poly}(n)$ for any $k\geq 3$ based on disentanglers\footnote{Our result thus indicates that the $1/2$ barrier of~\cite{jeronimo2024dimension} is not a barrier for $\QMAplus(3)$. It is a barrier for $\QMAplus(3)$ with only one nonnegative proof}. See also~\Cref{tab:amplification-comparison}.

\begin{table}[t]
\centering
\small
\setlength{\tabcolsep}{5pt}
\renewcommand{\arraystretch}{1.15}
\begin{tabularx}{\linewidth}{
  @{} l c >{\raggedright\arraybackslash}X c c @{}
}
\toprule
Result
& Proofs
& Proof restriction
& Completeness
& Soundness \\
\midrule
Jeronimo--Wu~\cite{jeronimo2023power}
& \(2\)
& Both proofs have nonnegative amplitudes
& \(1-\exp(-\operatorname{poly}(n))\)
& \(\frac79+\exp(-\operatorname{poly}(n))\) \\

Jeronimo--Wu~\cite{jeronimo2024dimension}
& \(3\)
& One proof has nonnegative amplitudes; the other two are unrestricted
& \(1-\exp(-\operatorname{poly}(n))\)
& \(\frac12+\frac1{\operatorname{poly}(n)}\) \\

This work
& \(2\)
& Both proofs have nonnegative amplitudes
& \(1-\frac1{\mathrm{poly}(n)}\)
& \(\frac14+\frac1{\mathrm{poly}(n)}\) \\
\bottomrule
\end{tabularx}
\caption{
Known gap-amplification results involving nonnegative-amplitude
unentangled proofs. 
}
\label{tab:amplification-comparison}
\end{table}

The value \(1/4\) is also the threshold reached by a general
sign-removal reduction.  Let \(\QMAreal(2)\) denote the two-proof model
in which both honest and adversarial proofs are restricted to real
amplitudes.  Proposition~\ref{prop:sign-removal} in Section~\ref{sec:preliminaries} shows that a verifier with soundness \(s\) against nonnegative product proofs has soundness at
most \(\min\{1,4s\}\) against arbitrary real product proofs. Hence crossing the line \(c=4s\) by an inverse-polynomial amount would have a strong complexity-theoretic consequence, as shown in the following corollary.

\par\medskip
\begin{corollary}
\label{cor:crossing-threshold}
Let \(\delta:\mathbb N\to[0,1]\) and
\(\eta:\mathbb N\to[0,1/4]\) be polynomial-time computable.
If
\[
  \NEXP
  \subseteq
  \QMAplus\!\left(2,1-\delta(n),\frac14-\eta(n)\right)
\]
and \(4\eta(n)-\delta(n)\ge1/q(n)\) for some polynomial \(q\) and all
sufficiently large \(n\), then \(\QMAreal(2)=\NEXP\).
\end{corollary}

The consequence \(\QMAreal(2)=\NEXP\) should be treated carefully. 
It leads to at least one of
\begin{itemize}
    \item \(\mathsf{QMA}(2) = \NEXP\), or
    \item \(\QMAreal(2)\neq \mathsf{QMA}(2) \).
\end{itemize}
The former is an important open question in quantum complexity theory. The latter is also significant, since, as stated in the survey~\cite{jeronimo2026qma}, it is believed that real and complex amplitudes should not fundamentally change the power of quantum verification models.  See also~\Cref{fig:gap-landscape} for the illustration of the complexity landscape of $\QMAplus(2)$ with respect to the gap.

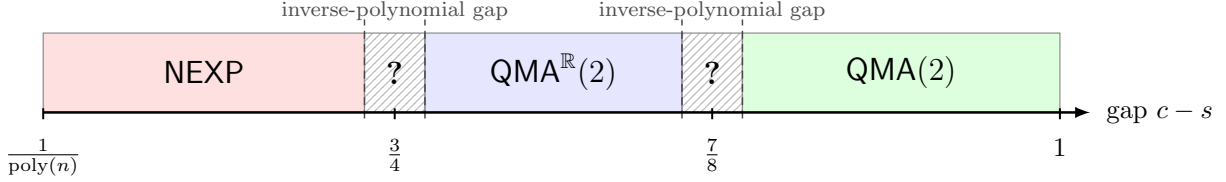
\begin{figure}[t]
\centering
\begin{tikzpicture}[
  x=1cm,
  y=1cm,
  classlabel/.style={font=\bfseries\large, text=black},
  intervallabel/.style={
    font=\scriptsize,
    align=center,
    text=black!82
  },
  boundary/.style={
    black!72,
    densely dashed,
    line width=0.55pt
  },
  axis/.style={
    -latex,
    line width=0.95pt
  },
]

\def\xstart{0.00}
\def\xnexpEnd{4.25}
\def\xgapAEnd{5.05}
\def\xrealEnd{8.45}
\def\xgapBEnd{9.25}
\def\xend{13.45}
\def\ybot{0.00}
\def\ytop{1.05}

\filldraw[
  fill=red!12,
  draw=black!45,
  line width=0.45pt
]
  (\xstart,\ybot) rectangle (\xnexpEnd,\ytop);

\filldraw[
  fill=blue!10,
  draw=black!45,
  line width=0.45pt
]
  (\xgapAEnd,\ybot) rectangle (\xrealEnd,\ytop);

\filldraw[
  fill=green!13,
  draw=black!45,
  line width=0.45pt
]
  (\xgapBEnd,\ybot) rectangle (\xend,\ytop);

\filldraw[
  pattern=north east lines,
  pattern color=black!28,
  draw=black!45,
  line width=0.45pt
]
  (\xnexpEnd,\ybot) rectangle (\xgapAEnd,\ytop);

\filldraw[
  pattern=north east lines,
  pattern color=black!28,
  draw=black!45,
  line width=0.45pt
]
  (\xrealEnd,\ybot) rectangle (\xgapBEnd,\ytop);

\foreach \x in {
  \xnexpEnd,
  \xgapAEnd,
  \xrealEnd,
  \xgapBEnd
}{
  \draw[boundary] (\x,-0.08) -- (\x,1.23);
}

\node[classlabel]
  at ({(\xstart+\xnexpEnd)/2},0.54)
  {$\mathsf{NEXP}$};

\node[classlabel]
  at ({(\xgapAEnd+\xrealEnd)/2},0.54)
  {$\mathsf{QMA}^{\mathbb R}(2)$};

\node[classlabel]
  at ({(\xgapBEnd+\xend)/2},0.54)
  {$\mathsf{QMA}(2)$};

\node[font=\bfseries\large]
  at ({(\xnexpEnd+\xgapAEnd)/2},0.52)
  {?};

\node[font=\bfseries\large]
  at ({(\xrealEnd+\xgapBEnd)/2},0.52)
  {?};

\draw[axis]
  (\xstart,\ybot) -- (13.85,\ybot)
  node[right=2pt,font=\small] {gap $c-s$};

\def\xthreefourths{4.65}
\def\xseveneighths{8.85}

\draw[line width=0.75pt]
  (\xthreefourths,0.055) -- (\xthreefourths,-0.09);

\draw[line width=0.75pt]
  (\xseveneighths,0.055) -- (\xseveneighths,-0.09);

\node[below=5pt,font=\small]
  at (\xthreefourths,-0.04)
  {$\frac34$};

\node[below=5pt,font=\small]
  at (\xseveneighths,-0.04)
  {$\frac78$};

\draw[line width=0.75pt]
  (\xstart,0.055) -- (\xstart,-0.09);

\draw[line width=0.75pt]
  (\xend,0.055) -- (\xend,-0.09);

\node[below=5pt,font=\small]
  at (\xstart,-0.04)
  {$\frac{1}{\mathrm{poly}(n)}$};

\node[below=5pt,font=\small]
  at (\xend,-0.04)
  {$1$};

\node[font=\scriptsize,text=black!70]
  at (\xthreefourths,1.35)
  {inverse-polynomial gap};

\node[font=\scriptsize,text=black!70]
  at (\xseveneighths,1.35)
  {inverse-polynomial gap};

\end{tikzpicture}

\caption{
A complexity landscape for $\QMAplus(2,c,s)$ as a function of the completeness--soundness gap \(c-s\), when $c=1-1/\mathrm{poly}$. The phase transition around $7/8$ is due to Proposition~\ref{prop:sign-removal}.
}
\label{fig:gap-landscape}
\end{figure}

\paragraph{A dimension-independent tensor-power approximation.}
The main technical ingredient in the proof is an approximation theorem for
symmetric states.

\begin{informaltheorem}[Informal version of
  Theorem~\ref{lem:tensor-power-approximation}]
Let \(\Gamma_m\) be a state supported on
\(\operatorname{Sym}^m(\mathcal H)\), and let \(\rho_\ell\) be its reduced
state on \(\ell\) registers.  For every fixed \(\eta>0\) and every
\[
  \ell\le \left(\frac12-\eta\right)\log_2 m,
\]
\(\rho_\ell\) is
\(\operatorname{poly}(\log m)m^{-\eta/2}\)-close in Hilbert--Schmidt norm
to a convex combination of states of the form
\[
  \bigl(\lvert\psi\rangle\langle\psi\rvert\bigr)^{\otimes\ell}.
\]
The error is independent of the local dimension \(\dim\mathcal H\).
\end{informaltheorem}

Jeronimo, Wu, and Xu~\cite{jeronimo2026optimal} recently proved a dimension-independent
quantum de Finetti theorem in Hilbert-Schmidt norm for the two-register reduced state, with the optimal worst-case rate \(\Theta(m^{-1/2})\).  Our result extends this
type of guarantee to reduced states on a number of registers that grows with \(m\), allowing as many as \((1/2-\eta)\log_2m\) registers. This distinction is essential for our
application: as discussed in \Cref{sec:intro-tensor-strategy}, our soundness analysis depends on reduced states on \(\Theta(\log m)\) registers, so the
two-register theorem alone does not suffice.

\subsection{Proof strategy for gap amplification}
Now we briefly explain our proof strategy. The proof follows the construction of the new verifier. First we turn the fixed-gap one-proof characterization of \(\NEXP\) into a robust
test on one larger register.
Next we ask two Merlins for long lists of such registers and enforce permutation symmetry.
The only substantial analytic question is then whether a small part of each symmetric proof
behaves sufficiently like a mixture of tensor powers.
Section~\ref{sec:tensor-power-approximation} proves the tensor-power
approximation theorem needed in this high-dimensional regime, but here we
use it only as a black box.

\paragraph{Step 1: obtain a robust one-proof test.}
We use
\[
  \QMAplus(k,c,s)\subseteq\NEXP
  =\QMAplus(1,c_0,s_0),
\]
where \(c_0>s_0\) are fixed constants, to replace the input verifier by
the fixed-gap one-proof verifier of Bassirian, Fefferman, and Marwaha
\cite{bassirian2024itcs}.  Run this verifier on \(L\) proof registers and accept
when at least a fixed fraction \(\theta\), with
\(s_0<\theta<c_0\), of the runs accept.  Regard these \(L\) registers
together as one larger register with Hilbert space \(\mathcal H\), and
let \(B\) be the acceptance operator of the threshold test.

On the honest tensor power, a Chernoff bound makes the rejection
probability exponentially small in \(L\).  Soundness does not assume
that the \(L\) registers are independent.  For a nonnegative proof,
fixing all but one register in the computational basis leaves a valid
nonnegative one-proof state.  The expected number of accepting runs is
therefore at most \(s_0L\), and Markov's inequality gives
\begin{equation}
\label{eq:intro-B-bound}
  \sup_{\substack{\ket{\psi}\ge0\\\lVert\ket{\psi}\rVert=1}}
  \bra{\psi}B\ket{\psi}
  \le r:=\frac{s_0}{\theta}<1.
\end{equation}
This is the only property of the starting protocol used by the
subsequent analytic argument.  The threshold construction is not a
generic amplification theorem for arbitrary one-proof \(\QMAplus\)
verifiers; it is available here because we first pass through the
fixed-gap characterization of \(\NEXP\).

\paragraph{Step 2: build the two-proof verifier.}
Each Merlin sends \(M\) registers with one-register space \(\mathcal H\).
The verifier projects each proof separately onto \(\Sym^M(\mathcal H)\).  
It then retains \(N=\Theta(\log M)\) registers from each proof, projects the resulting \(2N\) registers jointly onto \(\Sym^{2N}(\mathcal H)\), and applies \(B\) to every retained register.

If both Merlins send the honest tensor power, both symmetry tests
accept with certainty and all \(B\)-tests accept with probability
\(1-o(1)\).  A dishonest Merlin, however, may entangle the \(M\)
registers arbitrarily.  Projection onto \(\Sym^M(\mathcal H)\) makes
the proof permutation symmetric, but it does not make the registers
independent.

\paragraph{Step 3: prove soundness.}
The first symmetry test (the projection onto \(\Sym^M(\mathcal H)\)) leaves arbitrary correlations among the \(M\) registers. 
Theorem~\ref{lem:tensor-power-approximation} gives the approximation we
need: for every \(j\le(1/2-\eta)\log_2M\), the reduced state on \(j\)
registers is close, in Hilbert--Schmidt norm, to a mixture of identical
pure tensor powers.
Its error is independent of the one-register dimension $\dim\mathcal{H}$. This independence is essential because $\dim\mathcal{H}$ can be exponentially large.

Let \(\rho_N^{(1)}\) and \(\rho_N^{(2)}\) be the states on the \(N\)
registers selected from the two proofs, and define the subnormalized
operators
\[
  X=(B^{1/2})^{\otimes N}\rho_N^{(1)}(B^{1/2})^{\otimes N},
  \qquad
  Y=(B^{1/2})^{\otimes N}\rho_N^{(2)}(B^{1/2})^{\otimes N}.
\]
These are the operators associated with acceptance of all \(B\)-tests
on the selected registers.  The probability that the joint symmetry
test and all \(B\)-tests accept is then
\[
  \Tr\!\left[P_{\mathrm{sym}}^{(2N)}(X\otimes Y)\right]
  =\frac1{\binom{2N}{N}}
   \sum_{j=0}^{N}\binom Nj^2\Tr(X_jY_j),
\]
where \(X_j\) and \(Y_j\) are the \(j\)-register reduced operators of
\(X\) and \(Y\). The coefficient 
\[
\frac{\binom Nj^2}{\binom{2N}{N}}
\]
is largest when $j$ is close to $N/2$, so the dominant terms involve reduced states on $\Theta(N)$ registers.  Since the protocol takes $N=\Theta(\log M)$, its soundness analysis requires a dimension-independent approximation for reduced states on $\Theta(\log M)$ registers. 

We now analyze each of the two proofs separately. Let \(\rho_j\) denote its reduced state on \(j\) registers before applying the \(B\)-tests.  The tensor-power approximation gives
\[
  \sigma_j
  =
  \int
  \bigl(\lvert\phi\rangle\langle\phi\rvert\bigr)^{\otimes j}
  \,d\mu(\phi),
  \qquad
  \lVert\rho_j-\sigma_j\rVert_2=o(1).
\]
Because the original proof has nonnegative amplitudes and the
projection onto \(\operatorname{Sym}^M(\mathcal H)\) has nonnegative
matrix entries in the computational basis, \(\rho_j\) is entrywise
nonnegative.

By the symmetric-subspace identity above and the Hilbert--Schmidt
Cauchy--Schwarz inequality, it suffices to bound, separately for each
proof,
\[
  \left\|
    (B^{1/2})^{\otimes j}
    \rho_j
    (B^{1/2})^{\otimes j}
  \right\|_2^2.
\]
Replacing \(\rho_j\) by \(\sigma_j\) changes this bound by only \(o(1)\).
We are therefore reduced to bounding
\[
  \left\|
    (B^{1/2})^{\otimes j}
    \sigma_j
    (B^{1/2})^{\otimes j}
  \right\|_2^2
  =
  \iint
  |\langle\phi|B|\psi\rangle|^{2j}
  \,d\mu(\phi)d\mu(\psi).
\]

The states \(\lvert\phi\rangle\) appearing in the mixture need not have
nonnegative amplitudes, so \eqref{eq:intro-B-bound} cannot be applied
to them directly.  Entrywise nonnegativity of \(\rho_j\) supplies the
missing constraint.  For
\(\lvert w\rangle=\sum_z w_z\lvert z\rangle\), define
\[
  \lvert\operatorname{abs}(w)\rangle
  =
  \sum_z|w_z|\lvert z\rangle.
\]
Then
\[
  \langle w\rvert^{\otimes j}\rho_j
  \lvert w\rangle^{\otimes j}
  \le
  \langle\operatorname{abs}(w)\rvert^{\otimes j}
  \rho_j
  \lvert\operatorname{abs}(w)\rangle^{\otimes j}.
\]
The same inequality holds for \(\sigma_j\), up to the \(o(1)\)
approximation error.

We now explain how this comparison yields the constant \(1/4\).  Fix a
sufficiently small constant \(\gamma>0\), and define
\[
  \mathcal G_\gamma
  =
  \left\{
    \lvert\phi\rangle:
    \langle\phi|B|\phi\rangle\ge1-\gamma
  \right\},
  \qquad
  \lambda=\mu(\mathcal G_\gamma).
\]
For a fixed \(\lvert\phi\rangle\), write
\[
  I(\phi)
  :=
  \int
  |\langle\phi|B|\psi\rangle|^{2j}
  \,d\mu(\psi).
\]
If \(\lvert\phi\rangle\notin\mathcal G_\gamma\), the Cauchy--Schwarz
inequality for the positive operator \(B\) gives
\[
  I(\phi)\le(1-\gamma)^j.
\]

Now suppose that \(\lvert\phi\rangle\in\mathcal G_\gamma\), and let
\(\lvert x_\phi\rangle\) be the normalized nonnegative state obtained
by taking the absolute values of the coefficients of
\(B\lvert\phi\rangle\).  Writing
\(\langle\phi|B|\psi\rangle=\langle w|\psi\rangle\) with
\(\lvert w\rangle=B\lvert\phi\rangle\), we have
\(I(\phi)=\langle w\rvert^{\otimes j}\sigma_j\lvert w\rangle^{\otimes j}\).
Exchanging \(\sigma_j\) for \(\rho_j\) at a cost of \(o(1)\),
applying the absolute-value comparison above to \(\lvert w\rangle\),
whose entrywise absolute value is
\(\lVert w\rVert\,\lvert x_\phi\rangle\), and exchanging back, we
obtain
\[
  I(\phi)
  \le
  \int|\langle x_\phi|\psi\rangle|^{2j}\,d\mu(\psi)
  +o(1),
\]
the corresponding moment centered at \(\lvert x_\phi\rangle\); here
we also used \(\lVert w\rVert\le1\).  Since
\(\lvert x_\phi\rangle\) is nonnegative,
\eqref{eq:intro-B-bound} gives
\[
  \langle x_\phi|B|x_\phi\rangle\le r.
\]
Every \(\lvert\psi\rangle\in\mathcal G_\gamma\), on the other hand, is
almost unchanged by \(B\), in the sense that
\(\lVert(I-B)\lvert\psi\rangle\rVert\le\sqrt\gamma\).  Splitting
\(\langle x_\phi|\psi\rangle
  =\langle x_\phi|B|\psi\rangle
  +\langle x_\phi|(I-B)|\psi\rangle\)
and applying the Cauchy--Schwarz inequality to each term, we get
\[
  |\langle x_\phi|\psi\rangle|
  \le
  \lVert B\lvert x_\phi\rangle\rVert
  +\lVert(I-B)\lvert\psi\rangle\rVert
  \le
  \sqrt r+\sqrt\gamma
  =:\kappa,
\]
where the first term is bounded using \(B^2\preceq B\) and the bound
above.  For \(\gamma<(1-\sqrt r)^2\), the constant \(\kappa\) is
smaller than one.  Hence
\[
  I(\phi)
  \le
  (1-\lambda)+\lambda\kappa^{2j}+o(1).
\]

Integrating separately over
\(\mathcal G_\gamma\) and its complement gives
\[
  \left\|
    (B^{1/2})^{\otimes j}
    \sigma_j
    (B^{1/2})^{\otimes j}
  \right\|_2^2
  \le
  (1-\lambda)(1-\gamma)^j
  +\lambda(1-\lambda)
  +\lambda^2\kappa^{2j}
  +o(1).
\]
The first and third terms vanish as \(j\) grows, while
\[
  \lambda(1-\lambda)\le\frac14.
\]
Thus the squared Hilbert--Schmidt norm above is at most
\(1/4+o(1)\).

Applying this bound separately to the two proofs and using the
Hilbert--Schmidt Cauchy--Schwarz inequality gives
\[
  \operatorname{Tr}(X_jY_j)\le\frac14+o(1)
\]
for every \(j\) carrying nonnegligible weight in the
symmetric-subspace identity.  Averaging these bounds in that identity
gives total soundness \(1/4+o(1)\).

\subsection{Proof strategy for the tensor-power approximation theorem}
\label{sec:intro-tensor-strategy}

The approximation theorem used in Step~3 is a separate
information-theoretic result. Consider applying the optimal trace-distance theorem of Jeronimo, Wu, and
Xu~\cite{jeronimo2026optimal} directly: its error
\(O(\ell\sqrt d/m)\) grows with the local dimension
\(d=\dim\mathcal H\), and in our application \(d\) is exponentially
large while \(m\) is polynomial.  This dependence is not an artifact:
their lower bounds show that in trace norm the factor \(\sqrt d\) is
unavoidable.  A dimension-free statement is therefore only possible
in a weaker norm, and the Hilbert--Schmidt norm turns out to be
exactly weak enough.

The proof rests on two observations.  The first is an effective
dimension bound.  Since the one-register reduced state \(\rho_1\) has
unit trace, at most \(m\) of its eigenvalues can exceed \(1/m\).
Splitting the spectrum of \(\rho_1\) at \(1/m\) therefore yields a
\emph{heavy} subspace of dimension at most \(m\), and a \emph{light}
complement whose dimension may be huge but whose eigenvalues are all
at most \(1/m\).  On registers limited to the heavy subspace, the trace-distance theorem applies with the effective dimension \(m\) in place of \(d\): decomposing the input \(\Gamma_m \in \operatorname{Sym}^m(\mathcal H)\) according to the number of
registers in the heavy subspace and applying it to each component
gives a combined error of \(O(\ell/\sqrt m)\), with no reference to
\(d\).

The second observation is the new ingredient, and it explains why the
light part can be ignored.  It rests on an elementary fact.  If a
positive semidefinite operator \(X\) has all its eigenvalues at most
\(1/m\) and trace at most one, then
\[
  \Tr(X^2)\le\frac1m\,\Tr(X)\le\frac1m.
\]
The squared Hilbert--Schmidt norm of \(X\) is exactly \(\Tr(X^2)\),
the quantity known as the \emph{purity} of \(X\): for a normalized
state it equals one exactly when the state is pure, and it is small
when the spectrum is flat.  Thus a flat spectrum forces a small
Hilbert--Schmidt norm no matter how large the rank is.  The trace
norm offers no such discount.
To exploit this, let \(P\) and \(Q\) be the projectors onto the heavy
and light subspaces.  Inserting \(P\) or \(Q\) on each of the
\(\ell\) registers cuts \(\rho_\ell\) into \(2^\ell\) blocks, and our
goal is to show that every block containing at least one \(Q\) is
negligible.  This is the content of our bound
(Lemma~\ref{lem:projected-pattern-bound}): using the permutation
invariance of \(\Gamma_m\) and the swap identity from
Section~\ref{sec:preliminaries}, the purity of any such block turns
out to be at most
\[
  \Tr[(Q\rho_1Q)^2]+O(\ell^2/m).
\]
The first term involves only the light part of the one-register
state, so the elementary fact above bounds it by \(1/m\).  Hence
every block containing a \(Q\) has Hilbert--Schmidt norm
\(O(\ell/\sqrt m)\), independently of \(d\).

The theorem follows in three steps. 
First, we delete all blocks that contain a \(Q\). 
There are \(2^\ell-1\) of them, each of Hilbert--Schmidt size \(O(\ell/\sqrt m)\), so the deletion is affordable precisely when \(2^\ell\ll\sqrt m\). 
This is the origin of the range \(\ell\le(1/2-\eta)\log_2m\). 
Second, the surviving block lives entirely in the heavy subspace, where the effective
dimension is \(m\), and the first observation approximates it by a tensor-power mixture with error \(O(\ell/\sqrt m)\). 
Third, the deletion removed some trace, and we restore it by adding a mixture of tensor powers of light eigenvectors with exactly the missing weight, at a cost of \(O(\sqrt{\ell/m})\). 
In summary, the dimension \(d\) never enters.  The heavy subspace has dimension at most \(m\) and is covered by the existing trace-distance theorem, while every block that touches the light subspace has purity \(O(\ell^2/m)\) and is therefore negligible in Hilbert--Schmidt norm.

\subsection{Related work}
The analysis of Harrow and Montanaro's product test is sharpened and generalized
by Soleimanifar and Wright~\cite{soleimanifar2022testing}, and Beckey,
Jeronimo, and Wu~\cite{beckey2026optimal} recently determined the exact acceptance probability of the product test.

Restrictions on the allowed proofs lead to a different landscape.
Grilo, Kerenidis, and Sikora~\cite{grilo2015qma} showed that requiring an honest witness to
be a subset state does not change QMA or \(\QMAtwo\). 
By contrast, \(\QMAplus(k)\) restricts both honest and dishonest proofs to nonnegative amplitudes. Related recent models restrict the internal separability of a proof~\cite{bassirian2024quantum}, non-collapsing measurements~\cite{bassirian2025superposition}, or study how unentanglement
interacts with post-measurement branching in interactive proofs~\cite{grewal2025unentanglement}.
Dimension-independent disentanglers have also been used for error reduction in the pure quantum polynomial hierarchy~\cite{grewal2026pure}. 
A recent survey gives a broader context of the rapidly developing \(\QMAtwo\) landscape~\cite{jeronimo2026qma}.

Recent work on unentangled stoquastic proof systems is close in motivation, but concerns a different restriction.  
In \(\QMAplus(k)\), the proofs have nonnegative amplitudes while the verifier
is an arbitrary quantum circuit.  In \(\mathsf{StoqMA}(2)\), the verification procedure itself has stoquastic structure. 
Liu and Wu initiated a systematic study of \(\mathsf{StoqMA}(2)\)~\cite{liu2026unentangled}. Grilo and Rozos identified the separable stoquastic sparse Hamiltonian problem as a natural \(\mathsf{StoqMA}(2)\)-complete problem \cite{grilo2026complexity}. 
Gay and Jeronimo~\cite{gay2026collapse} subsequently proved \(\mathsf{StoqMA}(k)=\mathsf{StoqMA}\) for polynomially many proofs, using a positive, value-based de Finetti theorem for separately symmetric extensions.
Their theorem controls the optimum of an entrywise nonnegative test, whereas our approximation theorem controls the reduced state itself in Hilbert--Schmidt norm and is
designed for an unrestricted verifier. 
The two results therefore use related positivity and symmetry ideas but neither directly substitutes for the other.

The de Finetti results most directly comparable with our approximation theorem are discussed in Section~\ref{sec:tensor-power-approximation}, immediately before its proof.

\paragraph{Acknowledgments}
The author is grateful to Satoya Imai for helpful discussions regarding Ref.~\cite{marconi2026entanglement}, and to Carlo Marconi for presenting that work, which inspired the use of the symmetric projector for gap amplification. The author was supported by JSPS KAKENHI Grant Number JP25K24383 and MEXT AI for Science program (SPReAD-1000).

ChatGPT 5.5 and 5.6 were used to improve and simplify soundness analysis with substantial human input, improve writing, and in aiding the calculations.
The author takes full responsibility for the content of the manuscript.

\section{Preliminaries}
\label{sec:preliminaries}

All Hilbert spaces are finite dimensional, and a computational basis
is fixed whenever nonnegative amplitudes are discussed.  A pure state
\[
  \ket{\psi}=\sum_z \psi_z\ket{z}
\]
has \emph{nonnegative amplitudes} if \(\psi_z\in\mathbb R_{\ge0}\) for every computational-basis string \(z\). We use \(\ket{\psi}\ge0\) as shorthand for this condition.

For a Hilbert space \(\mathcal H\), let \(\D(\mathcal H)\) denote its density operators.
The subspace of \(\mathcal H^{\otimes m}\) invariant under every permutation of the \(m\) registers is denoted by \(\Sym^m(\mathcal H)\), and its orthogonal projector is
\[
  P_{\mathrm{sym}}^{(m)}
  =\frac1{m!}\sum_{\pi\in S_m}U_\pi,
\]
where the sum is over all permutations.
The acceptance operator of a verifier is the POVM element \(0\preceq A\preceq I\) whose expectation value is the verifier's acceptance probability.

\paragraph{Norms and inner products.}

For a linear operator \(X\) on a finite-dimensional Hilbert space,
\(\lVert X\rVert_1=\Tr\sqrt{X^\dagger X}\) denotes the trace norm,
\(\lVert X\rVert_2=\sqrt{\Tr(X^\dagger X)}\) the Hilbert--Schmidt
norm, and \(\lVert X\rVert_\infty\) the operator norm. 
These are the \(\ell_1\), \(\ell_2\), and \(\ell_\infty\) norms of the vector of
singular values of \(X\), so that
\[
  \lVert X\rVert_\infty
  \le\lVert X\rVert_2
  \le\lVert X\rVert_1.
\]
For positive semidefinite \(X\) we have
\(\lVert X\rVert_1=\Tr X\) and thus subnormalized positive
semidefinite operators have all three norms at most one. 
The Hilbert--Schmidt norm is induced by the inner product
\(\langle A,B\rangle_{\mathrm{HS}}=\Tr(A^\dagger B)\), whose
Cauchy--Schwarz inequality
\[
  |\Tr(A^\dagger B)|
  \le\lVert A\rVert_2\lVert B\rVert_2
\]
is used repeatedly throughout this paper, as is H\"older duality in the form
\(|\Tr(AB)|\le\lVert A\rVert_\infty\lVert B\rVert_1\). 
For a density operator \(\sigma\), the latter implies
\(|\Tr(A\sigma)|\le\lVert A\rVert_\infty\). 
All three norms are multiplicative under tensor products, and
\(\lVert\,\ket u\!\bra v\,\rVert_2
 =\lVert\ket u\rVert\,\lVert\ket v\rVert\)
for rank-one operators, and hence
\(\lVert(\proj w)^{\otimes j}\rVert_2=\lVert\ket w\rVert^{2j}\).

\paragraph{The swap identity.}

Let \(F\) be the swap operator on
\(\mathcal H\otimes\mathcal H\), defined by
\(F(\ket\alpha\otimes\ket\beta)=\ket\beta\otimes\ket\alpha\).  For all
operators \(A\) and \(B\) on \(\mathcal H\),
\[
  \Tr[(A\otimes B)F]
  =\Tr(AB).
\]
We refer to this as the \emph{swap identity}. 
For self-adjoint
\(X\) it gives
\(\lVert X\rVert_2^2=\Tr[(X\otimes X)F]\), which turns
Hilbert--Schmidt norms of reduced states into expectations of swap
operators on two copies of the underlying state. This is the form
used in the proof of \Cref{lem:tensor-power-approximation}.

\paragraph{Complexity classes and their relations.}

\begin{definition}[\(\QMAplus(k,c,s)\)]
A promise problem
\(\Pi=(\Pi_{\mathrm{yes}},\Pi_{\mathrm{no}})\) is in
\(\QMAplus(k,c,s)\) if there is a uniform polynomial-time quantum verifier receiving \(k(n)\) mutually unentangled proof states, each on polynomially many qubits, such that:
\begin{itemize}

\item if \(x\in\Pi_{\mathrm{yes}}\), there are \(k(n)\) pure proof
states with nonnegative computational-basis amplitudes whose tensor
product is accepted with probability at least \(c(n)\);

\item if \(x\in\Pi_{\mathrm{no}}\), every
such product of pure proof states with nonnegative amplitudes is
accepted with probability at most \(s(n)\).

\end{itemize}
\end{definition}

Following Jeronimo and Wu~\cite{jeronimo2024dimension}, we write \(\QMAreal(2,c,s)\) for the analogous two-proof class in which both honest and adversarial proof states are required to have real amplitudes. We use \(\mathsf{QMA}(2,c,s)\) for the standard model with arbitrary
complex-amplitude proofs. 
The following direct simulation shows the relation between these two proof models that we need.

\begin{proposition}
\label{prop:complex-to-real}
For every \(c>s\),
\[
  \mathsf{QMA}(2,c,s)
  \subseteq
  \QMAreal\!\left(2,\frac{1+c}{2},\frac{1+s}{2}\right).
\]
Consequently, \(\QMAtwo\subseteq\QMAreal(2)\).
\end{proposition}

\begin{proof}
Write each complex proof as
\(\ket{\psi_j}=\ket{a_j}+i\ket{b_j}\) with real vectors
\(\ket{a_j}\) and \(\ket{b_j}\), and encode it by the real state
\[
  \ket{\widetilde\psi_j}
  =\ket0\ket{a_j}+\ket1\ket{b_j},
\]
adding one qubit per proof. The encoding is applied separately to the two proofs, so it preserves the product structure required of honest proofs. 
The new verifier acts as follows: it measures the added qubit of each proof in the basis
\[
  \ket{y_\pm}=\frac{\ket0\pm i\ket1}{\sqrt2};
\]
if the two outcomes differ, it accepts; if both outcomes are \(-\),
it runs the original verifier \(V\) on the two remaining states; and
if both are \(+\), it runs the entrywise complex conjugate
\(\overline V\) of \(V\).

Consider an arbitrary pair of real unit proofs.  Every real unit
vector on the enlarged register can be written as
\(\ket{\widetilde\phi_j}=\ket0\ket{a_j}+\ket1\ket{b_j}\) with real
\(\ket{a_j}\) and \(\ket{b_j}\) satisfying
\(\lVert a_j\rVert^2+\lVert b_j\rVert^2=1\).  Set
\(\ket{\phi_j}=\ket{a_j}+i\ket{b_j}\).

\emph{Outcome distribution and remaining states.}
Since \(\braket{y_\pm}{0}=1/\sqrt2\) and
\(\braket{y_\pm}{1}=\mp i/\sqrt2\),
\[
  \bigl(\bra{y_\pm}\otimes I\bigr)\ket{\widetilde\phi_j}
  =\frac{\ket{a_j}\mp i\ket{b_j}}{\sqrt2}.
\]
Because \(\langle a_j | b_j\rangle\) is real, the cross terms cancel in
\[
  \bigl\lVert\,\ket{a_j}\mp i\ket{b_j}\bigr\rVert^2
  =\lVert a_j\rVert^2+\lVert b_j\rVert^2
  =1.
\]
Hence, on each proof, the two outcomes occur with probability \(1/2\)
each, and the normalized remaining state after outcome \(-\) is
\(\ket{\phi_j}\), while after outcome \(+\) it is the entrywise
conjugate \(\ket{\overline{\phi_j}}\).  The joint proof state is the
product \(\ket{\widetilde\phi_1}\otimes\ket{\widetilde\phi_2}\) and
the two measurements act on disjoint registers, so the outcomes are
independent and each of the four outcome pairs occurs with
probability exactly \(1/4\).

\emph{Acceptance probability.}
Let \(p\) be the acceptance probability of \(V\) on
\(\ket{\phi_1}\otimes\ket{\phi_2}\).  Running \(\overline V\) on
\(\ket{\overline{\phi_1}}\otimes\ket{\overline{\phi_2}}\) conjugates
every amplitude of the computation of \(V\) on
\(\ket{\phi_1}\otimes\ket{\phi_2}\), so its acceptance probability is
\(\overline p=p\).  The two unequal-outcome branches accept with
probability one.  Hence the overall acceptance probability is exactly
\[
  \frac14\cdot1
  +\frac14\cdot1
  +\frac14\,p
  +\frac14\,p
  =\frac12+\frac p2.
\]
Since the honest encodings decode to the original witnesses
(\(\ket{\phi_j}=\ket{\psi_j}\)), while an arbitrary pair of real
proofs decodes to a product of complex unit vectors, to which the
soundness of \(V\) applies, completeness and soundness become
\((1+c)/2\) and \((1+s)/2\), respectively.
\end{proof}

In particular, we have the inclusion
\[
  \QMAtwo\subseteq\QMAreal(2).
\]
The unparameterized notations \(\QMAreal(2)\) and \(\QMAtwo\) denote
the unions over polynomial-time computable completeness and soundness
parameters separated by an inverse-polynomial gap.  We work with a
standard finite universal gate set closed under entrywise complex
conjugation.  Equivalently, one may approximate conjugated gates to
inverse-polynomial accuracy and absorb the resulting error into the
promise gap.

When \(k(n)\) and all proof lengths are polynomially bounded and
\(c(n)-s(n)\) is inverse polynomial, direct exponential-time
simulation gives the containment
\[
  \QMAplus(k,c,s)\subseteq\NEXP.
\]
The same simulation gives
\[
  \QMAtwo\subseteq\QMAreal(2)\subseteq\NEXP,
\]
where the first inclusion is Proposition~\ref{prop:complex-to-real} and the second
does not require amplification.

We also give the following relation, which is the source of the complexity-theoretic barrier at the line \(c=4s\).

\begin{proposition}[Sign and phase removal]
\label{prop:sign-removal}
\begin{align}
  \QMAplus(2,c,s)
  &\subseteq
  \QMAreal\!\left(2,c,\min\{1,4s\}\right),
  \label{eq:real-sign-removal}\\
  \QMAplus(2,c,s)
  &\subseteq
  \mathsf{QMA}\!\left(2,c,\min\{1,8s\}\right).
  \label{eq:complex-sign-removal}
\end{align}
\end{proposition}

\begin{proof}[Proof of the real bound \eqref{eq:real-sign-removal}]
Fix an input and let \(0\preceq\Pi_{\mathrm{acc}}\preceq I\) be the
acceptance operator of a $\QMAplus(2,c,s)$ verifier.
The soundness says that
\[
  \bra{u}\bra{v}\Pi_{\mathrm{acc}}\ket{u}\ket{v}\le s
\]
for all nonnegative unit vectors \(\ket{u},\ket{v}\).

First let \(\ket{a},\ket{c}\) be arbitrary real vectors, not
necessarily normalized.  Write their positive- and negative-amplitude
parts as
\[
  \ket{a}=\ket{a^+}-\ket{a^-},
  \qquad
  \ket{c}=\ket{c^+}-\ket{c^-},
\]
so that \(\ket{a^\pm},\ket{c^\pm}\) are entrywise nonnegative with
disjoint supports.  By the soundness,
$\lVert\Pi_{\mathrm{acc}}^{1/2}(\ket{u}\otimes\ket{v})\rVert
  \le\sqrt s\,\lVert\ket{u}\rVert\,\lVert\ket{v}\rVert$
for all entrywise nonnegative \(\ket{u},\ket{v}\).  Expanding
\[
  \ket{a}\otimes\ket{c}
  =\ket{a^+}\otimes\ket{c^+}
  -\ket{a^+}\otimes\ket{c^-}
  -\ket{a^-}\otimes\ket{c^+}
  +\ket{a^-}\otimes\ket{c^-}
\]
and applying the triangle inequality to the four nonnegative product
terms gives
\begin{align}
 \left\|\Pi_{\mathrm{acc}}^{1/2}(\ket{a}\otimes\ket{c})\right\|
 &\le
 \sum_{x,y\in\{+,-\}}
 \left\|\Pi_{\mathrm{acc}}^{1/2}
   (\ket{a^{x}}\otimes\ket{c^{y}})\right\|
 \notag\\
 &\le
 \sqrt{s}\,
 \bigl(\lVert\ket{a^+}\rVert+\lVert\ket{a^-}\rVert\bigr)
 \bigl(\lVert\ket{c^+}\rVert+\lVert\ket{c^-}\rVert\bigr)
 \notag\\
 &\le 2\sqrt{s}\,
 \lVert\ket{a}\rVert\lVert\ket{c}\rVert,
 \label{eq:real-vector-bound}
\end{align}
where the last step uses
\(\lVert\ket{a^+}\rVert+\lVert\ket{a^-}\rVert
  \le\sqrt2\bigl(\lVert\ket{a^+}\rVert^2
                +\lVert\ket{a^-}\rVert^2\bigr)^{1/2}
  =\sqrt2\,\lVert\ket{a}\rVert\)
(by Cauchy--Schwarz and the disjointness of supports), and likewise
for \(\ket{c}\). Therefore,
\[
  \bra{a}\bra{b}\Pi_{\mathrm{acc}}\ket{a}\ket{b}\le 4s
\]
\end{proof}
\begin{proof}[Proof of the complex bound \eqref{eq:complex-sign-removal}]
Replace \(\Pi_{\mathrm{acc}}\) by the real acceptance operator
\(\Pi_{\mathrm{acc},\mathbb R}
  :=(\Pi_{\mathrm{acc}}+\overline{\Pi_{\mathrm{acc}}})/2\),
implemented by running, with probability \(1/2\) each, the original
verifier or its entrywise complex conjugate.  Since
\(0\preceq\Pi_{\mathrm{acc},\mathbb R}\preceq I\) and every real proof
has the same acceptance probability under \(\Pi_{\mathrm{acc}}\) and
\(\Pi_{\mathrm{acc},\mathbb R}\), completeness and the soundness on nonnegative proofs are unchanged and in particular, \eqref{eq:real-vector-bound} holds for
\(\Pi_{\mathrm{acc},\mathbb R}\).

Write arbitrary normalized complex proofs as
\(\ket{\psi}=\ket{a}+i\ket{b}\) and
\(\ket{\phi}=\ket{c}+i\ket{d}\) with
\(\ket{a},\ket{b},\ket{c},\ket{d}\) real, and set
\(\alpha=\lVert\ket{a}\rVert\), \(\beta=\lVert\ket{b}\rVert\),
\(\gamma=\lVert\ket{c}\rVert\), \(\delta=\lVert\ket{d}\rVert\), so
that \(\alpha^2+\beta^2=\gamma^2+\delta^2=1\).  Then
\(\ket{\psi}\otimes\ket{\phi}=\ket{r}+i\ket{t}\) with the real
vectors
\[
  \ket{r}=\ket{a}\otimes\ket{c}-\ket{b}\otimes\ket{d},
  \qquad
  \ket{t}=\ket{a}\otimes\ket{d}+\ket{b}\otimes\ket{c}.
\]
Because \(\Pi_{\mathrm{acc},\mathbb R}\) is Hermitian with real
entries, it is a real symmetric matrix, so
\(\bra{x}\Pi_{\mathrm{acc},\mathbb R}\ket{y}
 =\bra{y}\Pi_{\mathrm{acc},\mathbb R}\ket{x}\)
for all real vectors \(\ket{x},\ket{y}\).  Expanding
\(\bigl(\bra{r}-i\bra{t}\bigr)
  \Pi_{\mathrm{acc},\mathbb R}
  \bigl(\ket{r}+i\ket{t}\bigr)\),
the two cross terms therefore cancel, and
\[
  \bra{\psi}\bra{\phi}\Pi_{\mathrm{acc},\mathbb R}
  \ket{\psi}\ket{\phi}
  =\bra{r}\Pi_{\mathrm{acc},\mathbb R}\ket{r}
   +\bra{t}\Pi_{\mathrm{acc},\mathbb R}\ket{t}.
\]
Each of \(\ket{r}\) and \(\ket{t}\) is a sum of two real product
vectors, so the triangle inequality and
\eqref{eq:real-vector-bound} (with
\(\Pi_{\mathrm{acc},\mathbb R}\)) give
\[
  \bigl\|\Pi_{\mathrm{acc},\mathbb R}^{1/2}\ket{r}\bigr\|
  \le
  \bigl\|\Pi_{\mathrm{acc},\mathbb R}^{1/2}
    (\ket{a}\otimes\ket{c})\bigr\|
  +\bigl\|\Pi_{\mathrm{acc},\mathbb R}^{1/2}
    (\ket{b}\otimes\ket{d})\bigr\|
  \le
  2\sqrt s\,(\alpha\gamma+\beta\delta),
\]
hence
\(\bra{r}\Pi_{\mathrm{acc},\mathbb R}\ket{r}
  \le4s(\alpha\gamma+\beta\delta)^2\), and likewise
\(\bra{t}\Pi_{\mathrm{acc},\mathbb R}\ket{t}
  \le4s(\alpha\delta+\beta\gamma)^2\).
Finally,
\[
  (\alpha\gamma+\beta\delta)^2
  +(\alpha\delta+\beta\gamma)^2
  =(\alpha^2+\beta^2)(\gamma^2+\delta^2)
   +4\alpha\beta\gamma\delta
  =1+4\alpha\beta\gamma\delta
  \le2,
\]
since \(\alpha\beta\le(\alpha^2+\beta^2)/2=1/2\) and likewise
\(\gamma\delta\le1/2\).  Therefore
\(\bra{\psi}\bra{\phi}\Pi_{\mathrm{acc},\mathbb R}
  \ket{\psi}\ket{\phi}\le8s\), which proves
\eqref{eq:complex-sign-removal} for normalized complex product
proofs.
\end{proof}

We also use the following known one-proof characterization.

\begin{theorem}[Bassirian--Fefferman--Marwaha \cite{bassirian2024itcs}]
\label{thm:bfm}
There are constants \(1>c_0>s_0>0\) such that
\[
  \NEXP=\QMAplus(1,c_0,s_0).
\]
\end{theorem}

\section{Repetition for the one-proof verifier}
\label{sec:one-proof}

Fix an \(\NEXP\) language and regard it as the promise problem \(\Pi\).
By Theorem~\ref{thm:bfm}, choose the corresponding one-proof verifier
family with completeness \(c_0\) and soundness \(s_0\).  For a fixed
input, let \(\mathcal W\) be its proof space and let \(A\) be its
acceptance operator.  Choose a constant \(\theta\) satisfying
\[
  s_0<\theta<c_0.
\]
We call a register with Hilbert space \(\mathcal W\) a
\(\mathcal W\)-register.  For an integer \(L\), let
\(\mathcal H:=\mathcal W^{\otimes L}\), and call a register with this
Hilbert space an \(\mathcal H\)-register.  Apply the one-proof
verification circuit separately to each of the \(L\)
\(\mathcal W\)-registers and accept the \(\mathcal H\)-register if at
least \(\lceil\theta L\rceil\) runs accept.  Let \(B\) be the accepting
POVM element of this threshold test.  Equivalently, measure
\(\{A,I-A\}\) on every \(\mathcal W\)-register, and
\[
  B=
  \sum_{\substack{S\subseteq[L]\\|S|\ge\lceil\theta L\rceil}}
  \bigotimes_{i=1}^L C_i(S),
  \qquad
  C_i(S)=
  \begin{cases}
    A,&i\in S,\\
    I-A,&i\notin S.
  \end{cases}
\]

\begin{lemma}
\label{lem:threshold}
The acceptance operator \(B\) has the following properties.
\begin{enumerate}[label=(\roman*)]
\item If \(\ket{h}\ge0\) satisfies
      \(\bra{h}A\ket{h}\ge c_0\), then
      \[
        \bra{h^{\otimes L}}B\ket{h^{\otimes L}}
        \ge 1-e^{-\Omega(L)}.
      \]
\item If the fixed input is a no instance, every normalized pure state
      \(\ket{\Psi}\in\mathcal H\) with
      nonnegative amplitudes, including a state that may be entangled
      among the \(L\) \(\mathcal W\)-registers,
      satisfies
      \[
        \bra{\Psi}B\ket{\Psi}
        \le r:=\frac{s_0}{\theta}<1.
      \]
\end{enumerate}
\end{lemma}

\begin{proof}
For part~(i), on the product state \(\ket{h}^{\otimes L}\), the \(L\)
acceptance outcomes are independent, so a Chernoff bound applies.

For part~(ii), fix a \(\mathcal W\)-register \(i\) and expand the proof
state with nonnegative amplitudes
in the computational basis of all other registers:
\[
  \ket{\Psi}=\sum_z \ket{\psi_{i,z}}\otimes\ket{z},
  \qquad \ket{\psi_{i,z}}\ge0.
\]
The vectors \(\ket{\psi_{i,z}}\) may be subnormalized.  The reduced
state on \(\mathcal W\)-register \(i\) is
\[
  \rho_i=\sum_z\proj{\psi_{i,z}}.
\]
By one-proof soundness and homogeneity,
\[
  \Tr(A\rho_i)
  =\sum_z\bra{\psi_{i,z}}A\ket{\psi_{i,z}}
  \le s_0\sum_z\lVert\ket{\psi_{i,z}}\rVert^2
  =s_0.
\]
If \(X\) is the number of accepting runs, then
\(\E X\le s_0L\), without any independence assumption.  Markov's
inequality therefore gives
\[
  \Prb\{X\ge\lceil\theta L\rceil\}
  \le\frac{s_0L}{\lceil\theta L\rceil}
  \le\frac{s_0}{\theta}.
\]
\end{proof}

This gives a fixed constant \(r<1\) for every pure state with nonnegative amplitudes in \(\mathcal H\), including states that may
be entangled among the \(L\) \(\mathcal W\)-registers, not only states of
the form \(\ket{\psi}^{\otimes L}\).

\section{The two-proof protocol and its analysis}
\label{sec:protocol}

We first state the soundness bound used in the protocol.  Its proof,
given later in this section, combines the tensor-power approximation
theorem from Section~\ref{sec:tensor-power-approximation} with an exact
identity for the symmetric-subspace projection.

\begin{proposition}[Soundness bound for the two-proof test]\label{prop:two-proof-soundness}
Let $\mathcal H$ have a fixed computational basis, and let
$0\preceq B\preceq I$ satisfy
\[
 \sup_{\substack{\ket{\psi}\ge0\\\|\ket{\psi}\|=1}}
 \langle\psi|B|\psi\rangle
 \le r<1.
\]
Let
\[
 \rho_M^{(1)},\rho_M^{(2)}
 \in\mathrm{D}\!\left(\operatorname{Sym}^M(\mathcal H)\right)
\]
be entrywise nonnegative density operators.  For sufficiently large
$M$, set
\[
 N=2\left\lfloor\frac14\log_2M\right\rfloor,
\]
and let $\rho_N^{(1)}$ and $\rho_N^{(2)}$ be their reductions to $N$
registers.  Then
\begin{align}\label{eq:two-proof-soundness}
\operatorname{Tr}\!\left[
 P_{\mathrm{sym}}^{(2N)}
 B^{\otimes2N}
 \left(\rho_N^{(1)}\otimes\rho_N^{(2)}\right)
\right]
\le
\frac14+
\operatorname{poly}(\log M)M^{-\Omega_r(1)}.
\end{align}
The implicit constants depend only on $r$.
\end{proposition}

\subsection{Verifier}

Let \(n\) denote the input length.  Choose constants
\(0<\alpha<1/3\), \(C\in\mathbb N\), and \(C'>0\), with \(C\)
sufficiently large, and set
\[
  M=(n+2)^C,\qquad
N=2\left\lfloor\frac14\log_2M\right\rfloor,\qquad
L=\left\lceil C'\log(n+2)\right\rceil
\]
We choose \(C\) so that \(t\ge1\) for every input length.  Each Merlin
sends a normalized pure proof state
\(\ket{\Psi_i}\in\mathcal H^{\otimes M}\), \(i\in\{1,2\}\), whose
computational-basis amplitudes are nonnegative.  Here
\(\mathcal H=\mathcal W^{\otimes L}\), so each proof consists of \(M\)
\(\mathcal H\)-registers.  The joint proof is
\(\ket{\Psi_1}\otimes\ket{\Psi_2}\), as required in the definition
of \(\QMAplus(2)\).

The verifier performs the following tests.

\begin{tcolorbox}[title=Our amplification protocol, breakable=true]
\begin{enumerate}[label=\arabic*.]
\item On each proof \(\ket{\Psi_i}\), perform the measurement
      \(\{P_{\mathrm{sym}}^{(M)},I-P_{\mathrm{sym}}^{(M)}\}\).
      Reject unless both measurements have outcome
      \(P_{\mathrm{sym}}^{(M)}\).
\item From each proof, select any fixed \(N\) \(\mathcal H\)-registers
      and discard the other \(M-N\).  The choice is immaterial after
      the first test.
\item On the resulting \(2N\) \(\mathcal H\)-registers, perform
      \(\{P_{\mathrm{sym}}^{(2N)},I-P_{\mathrm{sym}}^{(2N)}\}\).
      Reject unless the outcome is \(P_{\mathrm{sym}}^{(2N)}\).
\item On each of the \(2N\) selected \(\mathcal H\)-registers, perform
      the two-outcome measurement \(\{B,I-B\}\).  Accept if and only if
      every measurement has outcome \(B\). Here $B$ is the same as of~\Cref{lem:threshold}.
\end{enumerate}
\end{tcolorbox}

Since \(P_{\mathrm{sym}}^{(2N)}\) commutes with \(B^{\otimes2N}\),
steps 3 and 4 have joint acceptance operator
\(P_{\mathrm{sym}}^{(2N)}B^{\otimes2N}\).

\subsection{Completeness}

On a yes instance, let \(\ket{h}\ge0\) be the honest proof for the
one-proof verifier.  Both Merlins send
\[
  \bigl(\ket{h}^{\otimes L}\bigr)^{\otimes M}.
\]
The two local symmetric-subspace measurements and the joint
symmetric-subspace measurement return their symmetric outcomes with
probability one.  Lemma~\ref{lem:threshold}(i) gives
\[
  c_{\mathrm{new}}
  \ge\bigl(1-e^{-\Omega(L)}\bigr)^{2N}
  =1-o(1).
\]
Taking \(C'\) sufficiently large makes the completeness error smaller
than \(1/p(n)\) for any prescribed polynomial \(p\).

\subsection{Soundness}

Consider a no instance.  Condition on both measurements in step 1
returning their symmetric outcomes.  If this event has probability
zero, there is nothing to prove.  Because the two initial proofs are
unentangled and the two measurements act separately, the conditional
state has the form
\[
  \rho_M^{(1)}\otimes\rho_M^{(2)},
  \qquad
  \rho_M^{(i)}\in\D(\Sym^M(\mathcal H)).
\]
Both the initial proof states and \(P_{\mathrm{sym}}^{(M)}\) are
entrywise nonnegative.  Hence each \(\rho_M^{(i)}\), and every reduced
state obtained from it, is entrywise nonnegative.

Lemma~\ref{lem:threshold}(ii) supplies
the condition needed for Proposition~\ref{prop:two-proof-soundness}, which therefore bounds the conditional probability that steps 3 and 4 both
accept by
\[
  \frac14+
  \operatorname{poly}(\log M)M^{-\Omega_{r}(1)}.
\]
Multiplying by the probability of reaching the conditional branch can
only decrease this expression.  Choosing \(C\) sufficiently large
makes the error smaller than \(1/p(n)\) for any prescribed polynomial
\(p\).

\subsection{Proof of Proposition~\ref{prop:two-proof-soundness}}
\label{sec:technical-soundness}
Now we proceed to prove Proposition~\ref{prop:two-proof-soundness}.
We first prove the following lemma.

\begin{lemma}\label{lem:second-moment-new}
Let \(0\preceq B\preceq I\) satisfy
\[
  \sup_{\substack{\ket{\psi}\ge0\\
                   \lVert\ket{\psi}\rVert=1}}
  \bra{\psi}B\ket{\psi}\le r<1.
\]
Let \(j\ge1\), let
\(\rho\in\D(\mathcal H^{\otimes j})\) be entrywise nonnegative in the
fixed computational basis, and let \(\mu\) be a probability measure
on unit vectors such that
\[
  \sigma=\int\bigl(\proj{\phi}\bigr)^{\otimes j}\,d\mu(\phi)
\]
satisfies \(\lVert\rho-\sigma\rVert_2\le\varepsilon\).  Then
\begin{align}
 &\left\|
   (B^{1/2})^{\otimes j}\rho(B^{1/2})^{\otimes j}
  \right\|_2^2 \notag\\
 &\quad\le
 \frac14+
 \left(1-\frac{(1-\sqrt r)^2}{4}\right)^j
 +\left(\frac{1+\sqrt r}{2}\right)^{2j}
 +4\varepsilon.
\label{eq:second-moment-bound}
\end{align}
\end{lemma}

\begin{proof}
First, observe that both
\[
(B^{1/2})^{\otimes j}\rho(B^{1/2})^{\otimes j} \text{ and }
(B^{1/2})^{\otimes j}\sigma(B^{1/2})^{\otimes j}
\]
have trace at most 1, as $\operatorname{Tr}\left[(B^{1/2})^{\otimes j}\rho(B^{1/2})^{\otimes j}\right]\le \operatorname{Tr}[\rho]=1$, and therefore their Hilbert--Schmidt norms are also bounded by 1. Consequently,
\begin{align*}
 &\left|
 \left\|(B^{1/2})^{\otimes j}\rho(B^{1/2})^{\otimes j}\right\|_2^2
 -
 \left\|(B^{1/2})^{\otimes j}\sigma(B^{1/2})^{\otimes j}\right\|_2^2
 \right|                                                        \\
 &\qquad = 
 \left(
 \left\|(B^{1/2})^{\otimes j}\rho(B^{1/2})^{\otimes j}\right\|_2
 +
 \left\|(B^{1/2})^{\otimes j}\sigma(B^{1/2})^{\otimes j}\right\|_2
 \right) \cdot
 \left(
 \left\|(B^{1/2})^{\otimes j}\rho(B^{1/2})^{\otimes j}\right\|_2
 -
 \left\|(B^{1/2})^{\otimes j}\sigma(B^{1/2})^{\otimes j}\right\|_2
 \right)       
 \\
 &\qquad\le (1+1) \cdot \lVert\rho-\sigma\rVert_2
 \le2\varepsilon.
\end{align*}
For the approximating state $\sigma$, direct expansion gives
\begin{equation}
\label{eq:second-moment-kernel}
\left\|
 (B^{1/2})^{\otimes j}\sigma(B^{1/2})^{\otimes j}
\right\|_2^2
=\iint|\bra{\phi}B\ket{\psi}|^{2j}
  \,d\mu(\phi)d\mu(\psi).
\end{equation}
For \(\ket{w}=\sum_z w_z\ket z\), write
\[
  \ket{\operatorname{abs}(w)}=\sum_z|w_z|\ket z.
\]
Entrywise nonnegativity of \(\rho\) and the triangle inequality give
\begin{align*}
  \bra{w^{\otimes j}}\rho\ket{w^{\otimes j}} =
  |\bra{w^{\otimes j}}\rho\ket{w^{\otimes j}}|
  &\le \sum_{z,z'}|w_z|\cdot|w_{z'}|\bra{z}\rho\ket{z'}\\
  &\le
  \bra{\operatorname{abs}(w)^{\otimes j}}
  \rho
  \ket{\operatorname{abs}(w)^{\otimes j}}.
\end{align*}

The Hilbert--Schmidt approximation of $\rho$ by $\sigma$,
the Hilbert--Schmidt Cauchy--Schwarz inequality, and
\[
 \left\|(\ket{w}\!\bra{w})^{\otimes j}\right\|_2
 =\lVert\ket{w}\rVert^{2j}
\]
give
\begin{align}\label{eq:absolute-value-moment}
\int |\langle w|\psi\rangle|^{2j}\,d\mu(\psi)
&=
(\bra{w})^{\otimes j}
\sigma
(\ket{w})^{\otimes j}
\notag\\
&=
(\bra{w})^{\otimes j}
\rho
(\ket{w})^{\otimes j}
+
(\bra{w})^{\otimes j}
(\sigma-\rho)
(\ket{w})^{\otimes j}
\notag\\
&\le
(\bra{w})^{\otimes j}
\rho
(\ket{w})^{\otimes j}
+
\left|
(\bra{w})^{\otimes j}
(\sigma-\rho)
(\ket{w})^{\otimes j}
\right|
\notag\\
&\le
(\bra{w})^{\otimes j}
\rho
(\ket{w})^{\otimes j}
+
\|\sigma-\rho\|_2
\left\|
(\ket{w}\!\bra{w})^{\otimes j}
\right\|_2
\notag\\
&\le
(\bra{w})^{\otimes j}
\rho
(\ket{w})^{\otimes j}
+
\varepsilon\lVert\ket{w}\rVert^{2j}
\notag\\
&\le
(\bra{\operatorname{abs}(w)})^{\otimes j}
\rho
(\ket{\operatorname{abs}(w)})^{\otimes j}
+
\varepsilon\lVert\ket{w}\rVert^{2j}
\notag\\
&=
(\bra{\operatorname{abs}(w)})^{\otimes j}
\sigma
(\ket{\operatorname{abs}(w)})^{\otimes j}
\notag\\
&\quad+
(\bra{\operatorname{abs}(w)})^{\otimes j}
(\rho-\sigma)
(\ket{\operatorname{abs}(w)})^{\otimes j}
+
\varepsilon\lVert\ket{w}\rVert^{2j}
\notag\\
&\le
(\bra{\operatorname{abs}(w)})^{\otimes j}
\sigma
(\ket{\operatorname{abs}(w)})^{\otimes j}
\notag\\
&\quad+
\left|
(\bra{\operatorname{abs}(w)})^{\otimes j}
(\rho-\sigma)
(\ket{\operatorname{abs}(w)})^{\otimes j}
\right|
+
\varepsilon\lVert\ket{w}\rVert^{2j}
\notag\\
&\le
(\bra{\operatorname{abs}(w)})^{\otimes j}
\sigma
(\ket{\operatorname{abs}(w)})^{\otimes j}
+
2\varepsilon\lVert\ket{w}\rVert^{2j}
\notag\\
&=
\int
|\langle\operatorname{abs}(w)|\psi\rangle|^{2j}
\,d\mu(\psi)
+
2\varepsilon\lVert\ket{w}\rVert^{2j}.
\end{align}
Here the sixth line uses the preceding inequality obtained from the
entrywise nonnegativity of $\rho$.  In the second last
line, we use
\[
 \lVert\ket{\operatorname{abs}(w)}\rVert
 =
 \lVert\ket{w}\rVert
\]
together with $\|\rho-\sigma\|_2\le\varepsilon$.

We now use \eqref{eq:absolute-value-moment} to bound the double integral in~\eqref{eq:second-moment-kernel}.
The strategy is to split the vectors according to whether their expectation under $B$
is close to one. 
A vector with smaller $B$-expectation gives an exponentially small contribution directly by Cauchy--Schwarz. For vectors with large $B$-expectation, we use~\eqref{eq:absolute-value-moment} to remove the phases from $B\ket{\varphi}$ and then apply the bound on nonnegative
vectors in the assumption of this lemma. 
The only remaining term will be the product of the measures of the two parts, which is at most $1/4$.

Define
\[
 \mathcal G
 =
 \left\{
  \ket{\varphi}:
  \langle\varphi|B|\varphi\rangle
  \ge
  1-\frac{(1-\sqrt r)^2}{4}
 \right\}.
\]

First suppose that $\ket{\varphi}\notin\mathcal G$.  For every unit
vector $\ket{\psi}$, Cauchy--Schwarz applied to
$B^{1/2}\ket{\varphi}$ and $B^{1/2}\ket{\psi}$ gives
\begin{align*}
 |\langle\varphi|B|\psi\rangle|^2
 &\le
 \langle\varphi|B|\varphi\rangle
 \langle\psi|B|\psi\rangle\\
 &\le
 1-\frac{(1-\sqrt r)^2}{4},
\end{align*}
where the second inequality uses $B\preceq I$.  Therefore
\[
 \int
 |\langle\varphi|B|\psi\rangle|^{2j}\,d\mu(\psi)
 \le
 \left(
  1-\frac{(1-\sqrt r)^2}{4}
 \right)^j.
\]

It remains to consider $\ket{\varphi}\in\mathcal G$.  Define the
nonnegative unit vector
\[
 \ket{x_\varphi}
 =
 \frac{
  \sum_z|\langle z|B|\varphi\rangle|\ket z
 }{
  \|B\ket{\varphi}\|
 }.
\]
We first bound the overlap of $\ket{x_\varphi}$ with any
$\ket{\psi}\in\mathcal G$.  Since $0\preceq B\preceq I$, we have
$(I-B)^2\preceq I-B$, and hence
\begin{align*}
 \|(I-B)\ket{\psi}\|^2
 &=
 \langle\psi|(I-B)^2|\psi\rangle\\
 &\le
 \langle\psi|(I-B)|\psi\rangle\\
 &=
 1-\langle\psi|B|\psi\rangle\\
 &\le
 \frac{(1-\sqrt r)^2}{4}.
\end{align*}
Thus
\[
 \|(I-B)\ket{\psi}\|
 \le
 \frac{1-\sqrt r}{2}.
\]

On the other hand, $\ket{x_\varphi}$ is a nonnegative unit vector, so
the hypothesis of the lemma applies to it.  Since $B^2\preceq B$,
\begin{align*}
 \|B\ket{x_\varphi}\|^2
 &=
 \langle x_\varphi|B^2|x_\varphi\rangle\\
 &\le
 \langle x_\varphi|B|x_\varphi\rangle\\
 &\le r.
\end{align*}
Consequently,
\begin{align*}
 |\langle x_\varphi|\psi\rangle|
 &\le
 |\langle x_\varphi|B|\psi\rangle|
 +
 |\langle x_\varphi|(I-B)|\psi\rangle|\\
 &\le
 \|B\ket{x_\varphi}\|
 +
 \|(I-B)\ket{\psi}\|\\
 &\le
 \sqrt r+\frac{1-\sqrt r}{2}\\
 &=
 \frac{1+\sqrt r}{2}.
\end{align*}

We now apply~\eqref{eq:absolute-value-moment} with $\ket w=B\ket{\varphi}$.  By the definition of
$\ket{x_\varphi}$, the vector obtained by taking the entrywise
absolute value of $B\ket{\varphi}$ is
$\|B\ket{\varphi}\|\ket{x_\varphi}$.  Hence
\begin{align*}
 \int
 |\langle\varphi|B|\psi\rangle|^{2j}\,d\mu(\psi)
 &\le
 \|B\ket{\varphi}\|^{2j}
 \left(
  \int
  |\langle x_\varphi|\psi\rangle|^{2j}\,d\mu(\psi)
  +2\varepsilon
 \right).
\end{align*}

For $\ket{\psi}\in\mathcal G$, the overlap in the last integral is at
most $(1+\sqrt r)/2$.  Outside $\mathcal G$, we use the trivial bound
$|\langle x_\varphi|\psi\rangle|\le1$.  It follows that
\begin{align*}
 \int
 |\langle x_\varphi|\psi\rangle|^{2j}\,d\mu(\psi)
 &\le
 1-\mu(\mathcal G)
 +
 \mu(\mathcal G)
 \left(\frac{1+\sqrt r}{2}\right)^{2j}.
\end{align*}
Since $\|B\ket{\varphi}\|\le1$, every
$\ket{\varphi}\in\mathcal G$ therefore satisfies
\begin{align*}
 \int
 |\langle\varphi|B|\psi\rangle|^{2j}\,d\mu(\psi)
 &\le
 1-\mu(\mathcal G)
 +
 \mu(\mathcal G)
 \left(\frac{1+\sqrt r}{2}\right)^{2j}
 +
 2\varepsilon.
\end{align*}

Finally, split the outer integral in~\eqref{eq:second-moment-kernel} into
$\ket{\varphi}\notin\mathcal G$ and
$\ket{\varphi}\in\mathcal G$.  The preceding two bounds give
\begin{align*}
&\iint
 |\langle\varphi|B|\psi\rangle|^{2j}
 \,d\mu(\varphi)d\mu(\psi)\\
&\le
 \bigl(1-\mu(\mathcal G)\bigr)
 \left(
  1-\frac{(1-\sqrt r)^2}{4}
 \right)^j\\
&\quad+
 \mu(\mathcal G)\bigl(1-\mu(\mathcal G)\bigr)\\
&\quad+
 \mu(\mathcal G)^2
 \left(\frac{1+\sqrt r}{2}\right)^{2j}
 +
 2\varepsilon\mu(\mathcal G)\\
&\le
 \frac14
 +
 \left(
  1-\frac{(1-\sqrt r)^2}{4}
 \right)^j
 +
 \left(\frac{1+\sqrt r}{2}\right)^{2j}
 +
 2\varepsilon.
\end{align*}
Here we used
\[
 \mu(\mathcal G)\bigl(1-\mu(\mathcal G)\bigr)\le\frac14,
 \qquad
 1-\mu(\mathcal G)\le1,
 \qquad
 \mu(\mathcal G)^2\le1.
\]
Combining it with~\eqref{eq:second-moment-kernel} and with the $2\varepsilon$ error incurred when replacing $\rho$ by $\sigma$, we obtain
\begin{align*}
\left\|
 (B^{1/2})^{\otimes j}
 \rho
 (B^{1/2})^{\otimes j}
\right\|_2^2
&\le
\left\|
 (B^{1/2})^{\otimes j}
 \sigma
 (B^{1/2})^{\otimes j}
\right\|_2^2
+2\varepsilon\\
&\le
\frac14
+
\left(
 1-\frac{(1-\sqrt r)^2}{4}
\right)^j
+
\left(\frac{1+\sqrt r}{2}\right)^{2j}
+
4\varepsilon,
\end{align*}
as claimed.
\end{proof}

We also use the following expression of the symmetric subspace projector.

\begin{lemma}
\label{lem:double-coset-new}
Let \(X,Y\) be subnormalized positive semidefinite operators supported
on \(\Sym^N(\mathcal H)\), and let \(X_j,Y_j\) be their reductions to
\(j\) registers.  Then
\begin{equation}
\label{eq:double-coset-new}
  \Tr[P_{\mathrm{sym}}^{(2N)}(X\otimes Y)]
  =\frac1{\binom{2N}{N}}
   \sum_{j=0}^N\binom Nj^2\Tr(X_jY_j).
\end{equation}
\end{lemma}

\begin{proof}
We evaluate the contribution of the permutations in
$P_{\mathrm{sym}}^{(2N)}$ in three steps.  First, we classify each
permutation by the number of registers that it moves between the two
groups of $N$ registers.  Second, we count the permutations in each
class.  Third, we evaluate one representative of each class using the
swap identity.

Expanding the symmetric-subspace projection gives
\begin{align*}
\operatorname{Tr}\!\left[
 P_{\mathrm{sym}}^{(2N)}(X\otimes Y)
\right]
=
\frac{1}{(2N)!}
\sum_{\pi\in S_{2N}}
\operatorname{Tr}\!\left[
 U_\pi(X\otimes Y)
\right].
\end{align*}

Because $X$ is supported on $\operatorname{Sym}^N(\mathcal H)$, every
permutation of its $N$ registers acts as the identity on the support
of $X$.  Thus, for every such permutation,
\[
 U_\pi X=XU_\pi=X.
\]
The same statement holds for $Y$.  Consequently, composing a
permutation on either side with permutations that act separately
within the two groups does not change its contribution to the trace.
More explicitly, if $\alpha,\beta\in S_N\times S_N$ permute registers
only within the two groups, then
\begin{align*}
\operatorname{Tr}\!\left[
 U_{\alpha\pi\beta}(X\otimes Y)
\right]
&=
\operatorname{Tr}\!\left[
 U_\alpha U_\pi U_\beta(X\otimes Y)
\right]\\
&=
\operatorname{Tr}\!\left[
 U_\alpha U_\pi(X\otimes Y)
\right]\\
&=
\operatorname{Tr}\!\left[
 U_\pi(X\otimes Y)U_\alpha
\right]\\
&=
\operatorname{Tr}\!\left[
 U_\pi(X\otimes Y)
\right].
\end{align*}

For $\pi\in S_{2N}$, let $j$ be the number of registers in the first
group whose images under $\pi$ lie in the second group.  Since $\pi$
is a bijection, exactly $j$ registers also move from the second group
to the first group.  Permutations with the same value of $j$ differ
only by relabelings within the two groups before and after the
permutation.  Hence their contributions to the trace are equal.

We now count the permutations having a fixed value of $j$.  The image
of the first group contains $N-j$ registers from the first group and
$j$ registers from the second group.  There are
\[
 \binom{N}{N-j}\binom{N}{j}
 =
 \binom{N}{j}^2
\]
ways to choose this image.  Once the image is fixed, there are $N!$
bijections from the first group onto it.  The second group must map
bijectively onto the complementary set of $N$ registers, which gives
another $N!$ choices.  Therefore, the number of permutations in the
class indexed by $j$ is
\[
 (N!)^2\binom{N}{j}^2.
\]
Its relative weight in the average over $S_{2N}$ is consequently
\[
 \frac{(N!)^2\binom{N}{j}^2}{(2N)!}
 =
 \frac{\binom{N}{j}^2}{\binom{2N}{N}}.
\]

It remains to evaluate the trace for one representative of this
class.  Choose the representative that swaps $j$ corresponding
registers between the two groups and fixes all remaining registers.

Let $\operatorname{SWAP}_j$ denote the operator that exchanges the two
groups of $j$ registers.  Since the representative permutation acts
as $\operatorname{SWAP}_j$ on the exchanged registers and as the
identity on all remaining registers, the defining property of the
partial trace gives
\begin{align*}
\operatorname{Tr}\!\left[
 U_\pi(X\otimes Y)
\right]
&=
\operatorname{Tr}\!\left[
 \operatorname{SWAP}_j
 \left(
  \operatorname{Tr}_{N-j}(X)
  \otimes
  \operatorname{Tr}_{N-j}(Y)
 \right)
\right]\\
&=
\operatorname{Tr}\!\left[
 \operatorname{SWAP}_j(X_j\otimes Y_j)
\right].
\end{align*}

To evaluate the remaining trace, let $\{\ket a\}$ be an orthonormal
basis of the $j$-register space.  The swap operator has the expansion
\[
 \operatorname{SWAP}_j
 =
 \sum_{a,b}
 \ket b\!\bra a\otimes\ket a\!\bra b.
\]
Therefore,
\begin{align*}
\operatorname{Tr}\!\left[
 \operatorname{SWAP}_j(X_j\otimes Y_j)
\right]
&=
\sum_{a,b}
 \operatorname{Tr}(\ket b\!\bra a X_j)\,
 \operatorname{Tr}(\ket a\!\bra b Y_j)\\
&=
\sum_{a,b}
 \langle a|X_j|b\rangle
 \langle b|Y_j|a\rangle\\
&=
\sum_a
 \left\langle a\left|
 X_j
 \left(\sum_b\ket b\!\bra b\right)
 Y_j
 \right|a\right\rangle\\
&=
\sum_a\langle a|X_jY_j|a\rangle\\
&=
\operatorname{Tr}(X_jY_j).
\end{align*}
For
$j=0$, we interpret $X_0=\operatorname{Tr}(X)$ and
$Y_0=\operatorname{Tr}(Y)$, so the same formula also covers the
identity permutation class.

Combining the number of permutations in each class with the value of
their trace contribution yields
\begin{align*}
\operatorname{Tr}\!\left[
 P_{\mathrm{sym}}^{(2N)}(X\otimes Y)
\right]
&=
\frac{1}{(2N)!}
\sum_{j=0}^N
 (N!)^2\binom{N}{j}^2
 \operatorname{Tr}(X_jY_j)\\
&=
\frac{1}{\binom{2N}{N}}
\sum_{j=0}^N
 \binom{N}{j}^2
 \operatorname{Tr}(X_jY_j),
\end{align*}
which completes the proof.
\end{proof}

\begin{proof}[Proof of Proposition~\ref{prop:two-proof-soundness}]
The proof has two parts.  First, we use
Lemma~\ref{lem:double-coset-new} to express the quantity
in~\eqref{eq:two-proof-soundness} as a weighted average of overlaps
between reduced operators, and we bound each overlap by the quantities
controlled by Lemma~\ref{lem:second-moment-new}.  Second, we show that
the weights are concentrated on values of $j$ for which the
tensor-power approximation required by
Lemma~\ref{lem:second-moment-new} has inverse-polynomial error.

\paragraph{The weighted-average formula.}
For $i\in\{1,2\}$, let $\rho_j^{(i)}$ denote the reduction of
$\rho_M^{(i)}$ to $j$ registers, and define
\[
 X^{(i)}
 =
 (B^{1/2})^{\otimes N}
 \rho_N^{(i)}
 (B^{1/2})^{\otimes N},
 \qquad
 X_j^{(i)}
 =
 \operatorname{Tr}_{N-j}\!\left(X^{(i)}\right).
\]

Because $\rho_M^{(i)}$ is supported on
$\operatorname{Sym}^M(\mathcal H)$, its reduction $\rho_N^{(i)}$ is
supported on $\operatorname{Sym}^N(\mathcal H)$.  Moreover,
$(B^{1/2})^{\otimes N}$ commutes with every permutation of the $N$
registers.  It follows that $X^{(i)}$ is also supported on
$\operatorname{Sym}^N(\mathcal H)$.  It is positive semidefinite and
subnormalized, since
\[
 \operatorname{Tr}\!\left(X^{(i)}\right)
 =
 \operatorname{Tr}\!\left(B^{\otimes N}\rho_N^{(i)}\right)
 \le
 \operatorname{Tr}\!\left(\rho_N^{(i)}\right)
 =
 1.
\]

The commutivity of $P_{\mathrm{sym}}^{(2N)}$ and
$(B^{1/2})^{\otimes 2N}$ and Lemma~\ref{lem:double-coset-new} gives
\begin{align}
&\operatorname{Tr}\!\left[
 P_{\mathrm{sym}}^{(2N)}
 \left(X^{(1)}\otimes X^{(2)}\right)
\right]\notag\\
&=
\operatorname{Tr}\!\left[
 (B^{1/2})^{\otimes 2N}
 P_{\mathrm{sym}}^{(2N)}
 (B^{1/2})^{\otimes 2N}
 \left(\rho_N^{(1)}\otimes\rho_N^{(2)}\right)
\right]\notag\\
&=
\operatorname{Tr}\!\left[
 P_{\mathrm{sym}}^{(2N)}
 B^{\otimes 2N}
 \left(\rho_N^{(1)}\otimes\rho_N^{(2)}\right)
\right]\notag\\
&=
\frac{1}{\binom{2N}{N}}
\sum_{j=0}^{N}
 \binom{N}{j}^{2}
 \operatorname{Tr}\!\left(
  X_j^{(1)}X_j^{(2)}
 \right).
\label{eq:prop8-average}
\end{align}

\paragraph{Bounding each summand.}
We next relate $X_j^{(i)}$ to the operator obtained by applying
$B^{1/2}$ directly to the $j$-register reduction $\rho_j^{(i)}$.
Let $K$ be any positive semidefinite operator on $j$ registers.  By
the definition of the partial trace,
\begin{align*}
\operatorname{Tr}\!\left(KX_j^{(i)}\right)
&=
\operatorname{Tr}\!\left[
 \left(K\otimes I^{\otimes(N-j)}\right)X^{(i)}
\right]\\
&=
\operatorname{Tr}\!\left[
 \left(
  (B^{1/2})^{\otimes j}
  K
  (B^{1/2})^{\otimes j}
 \right)
 \otimes B^{\otimes(N-j)}
 \;\rho_N^{(i)}
\right].
\end{align*}
Moreover,
\begin{align*}
&
\left(
 (B^{1/2})^{\otimes j}
 K
 (B^{1/2})^{\otimes j}
\right)
\otimes B^{\otimes(N-j)}
\\
&\qquad\preceq
\left(
 (B^{1/2})^{\otimes j}
 K
 (B^{1/2})^{\otimes j}
\right)
\otimes I^{\otimes(N-j)}.
\end{align*}
Therefore,
\begin{align*}
\operatorname{Tr}\!\left(KX_j^{(i)}\right)
&\le
\operatorname{Tr}\!\left[
 \left(
  (B^{1/2})^{\otimes j}
  K
  (B^{1/2})^{\otimes j}
 \right)
 \otimes I^{\otimes(N-j)}
 \;\rho_N^{(i)}
\right]\\
&=
\operatorname{Tr}\!\left[
 (B^{1/2})^{\otimes j}
 K
 (B^{1/2})^{\otimes j}
 \rho_j^{(i)}
\right]\\
&=
\operatorname{Tr}\!\left[
 K
 (B^{1/2})^{\otimes j}
 \rho_j^{(i)}
 (B^{1/2})^{\otimes j}
\right].
\end{align*}
This holds for every $K\succeq0$.  In particular, taking
$K=\ket v\!\bra v$ for an arbitrary vector $\ket v$ shows that
\[
 \bra v
 \left[
  (B^{1/2})^{\otimes j}
  \rho_j^{(i)}
  (B^{1/2})^{\otimes j}
  -
  X_j^{(i)}
 \right]
 \ket v
 \ge0.
\]
Hence
\begin{align}
0
\preceq
X_j^{(i)}
\preceq
(B^{1/2})^{\otimes j}
\rho_j^{(i)}
(B^{1/2})^{\otimes j}.
\label{eq:prop8-order}
\end{align}

We now use this operator inequality to bound the overlap appearing
in~\eqref{eq:prop8-average}.  Since $X_j^{(i)}$ is positive semidefinite, it follows from~\eqref{eq:prop8-order} that
\[
 \left\|X_j^{(i)}\right\|_2
 \le
 \left\|
  (B^{1/2})^{\otimes j}
  \rho_j^{(i)}
  (B^{1/2})^{\otimes j}
 \right\|_2.
\]
The Hilbert--Schmidt Cauchy--Schwarz inequality therefore gives
\begin{align}
\operatorname{Tr}\!\left(
 X_j^{(1)}X_j^{(2)}
\right)
&\le
\left\|X_j^{(1)}\right\|_2
\left\|X_j^{(2)}\right\|_2
\notag\\
&\le
\prod_{i=1}^{2}
\left\|
 (B^{1/2})^{\otimes j}
 \rho_j^{(i)}
 (B^{1/2})^{\otimes j}
\right\|_2.
\label{eq:prop8-hs-bound}
\end{align}

\paragraph{Averaging over $j$.}
Let $J$ be the hypergeometric random variable with
\[
 \Pr(J=j)
 =
 \frac{\binom{N}{j}^{2}}{\binom{2N}{N}},
 \qquad 0\le j\le N.
\]
Equivalently, $J$ counts the number of marked objects obtained when
$N$ objects are sampled uniformly without replacement from a
population of $2N$ objects containing $N$ marked objects.  In
particular, $\mathbb E[J]=N/2$, and
\eqref{eq:prop8-average} is the expectation of
$\operatorname{Tr}(X_J^{(1)}X_J^{(2)})$.  Hoeffding's inequality for
sampling without replacement gives
\begin{align}
 \Pr\!\left(J<\frac N4\ \text{or}\ J>\frac{3N}{4}\right)
 \le 2e^{-N/8}.
 \label{eq:prop8-hypergeometric-tail}
\end{align}

Consider a value of $j$ in the central range
\[
 \frac N4\le j\le\frac{3N}{4}.
\]
Since
\[
 N=2\left\lfloor\frac14\log_2 M\right\rfloor,
\]
we have
\[
 j\le\frac{3N}{4}
 \le\frac38\log_2 M
 =\left(\frac12-\frac18\right)\log_2 M.
\]
Theorem~\ref{lem:tensor-power-approximation}, applied to
$\rho_M^{(i)}$ with parameter $\eta=1/8$, therefore shows that each
$\rho_j^{(i)}$ has Hilbert--Schmidt distance at most
\[
 \operatorname{poly}(\log M)M^{-1/16}
\]
from a mixture of $j$-fold tensor powers.  Moreover,
$\rho_j^{(i)}$ is entrywise nonnegative, because it is a reduction of
$\rho_M^{(i)}$, and
\[
 j\ge\frac N4=\Omega(\log M).
\]
Since $r<1$, Lemma~\ref{lem:second-moment-new} now gives, uniformly
for $N/4\le j\le 3N/4$ and for both $i\in\{1,2\}$,
\begin{align}
 \left\|
  (B^{1/2})^{\otimes j}
  \rho_j^{(i)}
  (B^{1/2})^{\otimes j}
 \right\|_2^2
 \le
 \frac14+
 \operatorname{poly}(\log M)M^{-\Omega_r(1)}.
 \label{eq:prop8-central-second-moment}
\end{align}
Combining this bound with~\eqref{eq:prop8-hs-bound} yields
\begin{align}
 \operatorname{Tr}\!\left(X_j^{(1)}X_j^{(2)}\right)
 \le
 \frac14+
 \operatorname{poly}(\log M)M^{-\Omega_r(1)}
 \label{eq:prop8-central-overlap}
\end{align}
throughout the central range.

For every $0\le j\le N$, the operators $X_j^{(i)}$ are positive
semidefinite and subnormalized.  Hence
\[
 \operatorname{Tr}\!\left(X_j^{(1)}X_j^{(2)}\right)
 \le
 \left\|X_j^{(1)}\right\|_2
 \left\|X_j^{(2)}\right\|_2
 \le1.
\]
Using~\eqref{eq:prop8-central-overlap} in the central range and this
trivial bound outside it, equation~\eqref{eq:prop8-average} gives
\begin{align*}
&\operatorname{Tr}\!\left[
 P_{\mathrm{sym}}^{(2N)}
 B^{\otimes2N}
 \left(\rho_N^{(1)}\otimes\rho_N^{(2)}\right)
\right]\\
&\quad=
 \mathbb E\!\left[
  \operatorname{Tr}\!\left(X_J^{(1)}X_J^{(2)}\right)
 \right]\\
&\quad\le
 \left(
  \frac14+
  \operatorname{poly}(\log M)M^{-\Omega_r(1)}
 \right)
 \Pr\!\left(\frac N4\le J\le\frac{3N}{4}\right)
 +
 \Pr\!\left(J<\frac N4\ \text{or}\ J>\frac{3N}{4}\right)\\
&\quad\le
 \frac14+
 \operatorname{poly}(\log M)M^{-\Omega_r(1)}
 +2e^{-N/8}\\
&\quad=
 \frac14+
 \operatorname{poly}(\log M)M^{-\Omega_r(1)}.
\end{align*}
 This proves the proposition.
\end{proof}

\subsection{Efficient implementation and proof of Theorem~\ref{thm:main}}

For either \(R=M\) or \(R=2N\), the measurement
\[
  \left\{P_{\mathrm{sym}}^{(R)},I-P_{\mathrm{sym}}^{(R)}\right\}
\]
is efficiently implementable by weak Schur sampling.  Apply the
high-dimensional Schur transform, measure the partition label
\(\lambda\), accept if and only if \(\lambda=(R)\), and apply the
inverse Schur transform on the accepting branch.  The
high-dimensional Schur transform of Burchardt et
al.~\cite{burchardt2025high} has complexity \(\widetilde O(R^4)\) in
the regime relevant here, with only polylogarithmic dependence on the
local dimension and polynomial dependence on the logarithm of the
inverse circuit error $\delta_{Schur}$.  The circuit error $\delta_{Schur}$ can be chosen
inverse-exponentially small without affecting polynomial-time implementability.

Since \(R\) and \(\log\dim(\mathcal H)\) are polynomially bounded in
the input length, both symmetric-subspace measurements are implemented
by uniform polynomial-size quantum circuits.  Hence the whole
procedure, including \(B^{\otimes2N}\), runs in polynomial time.
Together with the completeness and soundness analyses, this proves
\[
  \NEXP
  \subseteq
  \QMAplus\!\left(2,1-\frac1{p(n)},\frac14+\frac1{p(n)}\right).
\]
The reverse containment follows by direct exponential-time simulation,
so Theorem~\ref{thm:main} follows.  Finally,
Corollary~\ref{cor:general-amplification} follows from
\(\QMAplus(k,c,s)\subseteq\NEXP\).

 
\section{A dimension-independent de Finetti theorem}
\label{sec:tensor-power-approximation}
 
This section proves the dimension-independent tensor-power approximation
used in the soundness analysis.  Let \(\Gamma_m\) be any state supported
on \(\Sym^m(\mathcal H)\).  We show that its reduced state on \(\ell\)
registers is close in Hilbert--Schmidt norm to a convex combination of
states of the form \((\proj\psi)^{\otimes\ell}\), with an error independent
of \(\dim\mathcal H\).  Statements of this form are commonly called
finite quantum de Finetti theorems.  We state and prove the precise form
needed here.
 
\begin{theorem}[Dimension-independent approximation by tensor-power mixtures]
\label{lem:tensor-power-approximation}
Let \(\Gamma_m\) be a density operator supported on
\(\Sym^m(\mathcal H)\), and let \(\rho_\ell\) be its reduced state on
\(1\le\ell\le m\) registers.  There is a probability measure \(\mu\)
on the unit sphere of \(\mathcal H\) such that
\[
  \sigma_\ell
  :=\int \bigl(\proj{\psi}\bigr)^{\otimes\ell}\,d\mu(\psi)
\]
satisfies
\begin{equation}
\label{eq:hs-tensor-power-main}
  \lVert\rho_\ell-\sigma_\ell\rVert_2
  \le
  \sqrt{2(2^\ell-1)}
  \left(
    \frac{3+2\ell(\ell-1)}m
  \right)^{1/4}
  +\frac{8\ell+\sqrt\ell}{\sqrt m}.
\end{equation}
In particular, for every fixed \(0<\eta<1/2\), if
\[
  \ell\le(1/2-\eta)\log_2m,
\]
then one can choose \(\sigma_\ell\) so that
\[
  \lVert\rho_\ell-\sigma_\ell\rVert_2
  \le \operatorname{poly}(\log m)m^{-\eta/2}.
\]
The bound is uniform in the local dimension
\(\dim\mathcal H\).
\end{theorem}
 
\medskip
\noindent\textbf{Relation to previous bounds.}
 
Theorem~\ref{lem:tensor-power-approximation} approximates a reduced
state by a mixture of identical pure tensor powers.  Standard finite
quantum de Finetti theorems give this conclusion in trace norm, with
an error growing linearly with the local dimension
\cite{konig2005finetti,christandl2007one,lewin2015remarks}.  For a state supported on the fully
symmetric subspace, Jeronimo, Wu, and Xu~\cite[Theorem~1.3]{jeronimo2026optimal} recently proved the optimal
trace-norm rate: for \(2\le\ell\le m\),
\[
  \inf_{\sigma_\ell}
  \frac12\lVert\rho_\ell-\sigma_\ell\rVert_1
  \le \frac{4\ell\sqrt{d-1}}m,
\]
where \(d=\dim\mathcal H\) and the infimum is over mixtures of
identical pure tensor powers.  Matching
lower bounds follow from explicit constructions going back
to~\cite{christandl2007one}; in particular, the linear dependence on \(\ell\)
cannot be improved in general \cite[Theorem~B.5]{jeronimo2026optimal}.
 
For the reduced state on two registers, they also prove the
dimension-independent Hilbert--Schmidt bound
\[
  \inf_{\sigma_2}\lVert\rho_2-\sigma_2\rVert_2
  \le
  \min\!\left\{
    \frac{\sqrt{d-1}}{m-1},
    \frac8{\sqrt{m-1}}
  \right\},
\]
whose \(m^{-1/2}\) rate is optimal when the local dimension may grow
with \(m\) \cite[Theorems~6.1 and~B.6]{jeronimo2026optimal}, the latter by a
construction adapted from Jiang, Tacla, and Caves~\cite{jiang2017bosonic}.
The distinguishing feature of Theorem~\ref{lem:tensor-power-approximation} is that the number of registers
under consideration may grow as large as \(\Theta(\log m)\).  The two-register theorem does not by itself yield this statement. 
In particular, grouping several original registers into one larger register would produce tensor powers of an arbitrary state on the grouped register, rather than tensor powers of one state on the original register.
 
A complementary route to dimension independence keeps arbitrary
permutation-invariant inputs but restricts the measurements against
which the error is evaluated.  Brand\~ao, Christandl, and Yard and Brand\~ao and Christandl~\cite{brandao2011faithful,brandao2012detection} showed that the dimension dependence improves to logarithmic against one-way LOCC measurements, and Brand\~ao and Harrow proved de Finetti theorems under local measurements with algorithmic applications~\cite{brandao2013quantum}. 
These bounds control restricted measurements only, whereas
Theorem~\ref{lem:tensor-power-approximation} controls the reduced state itself in Hilbert--Schmidt norm, as required by an unrestricted verifier.
 
An earlier Frobenius-norm (equivalently, Hilbert--Schmidt-norm)
theorem of Trimborn, Werner, and Witthaut~\cite{trimborn2016quantum} treats general
\(\ell\) for the same input class and the same approximating
mixtures, with the dimension-dependent bound
\[
  \inf_{\sigma_\ell}\lVert\rho_\ell-\sigma_\ell\rVert_2^2
  \le \frac{\ell^4(d^2+\ell)}{(m+d)^2}.
\]
To our knowledge, Theorem~\ref{lem:tensor-power-approximation} is the
first finite quantum de Finetti theorem in Hilbert--Schmidt norm that
simultaneously permits the number of registers under consideration to
grow unboundedly with \(m\), remains independent of the local
dimension, and approximates by mixtures of identical pure tensor
powers.
 
\medskip
\noindent\textbf{Proof overview.}
 
The proof has three steps.  We first split the one-register reduced
state at eigenvalue \(1/m\).  The part above this threshold has rank at
most \(m\), so the theorem of Jeronimo, Wu, and Xu applies there.  The
part at or below the threshold can have arbitrarily large rank, but its
contribution to the \(\ell\)-register reduced state is small in
Hilbert--Schmidt norm.  After discarding this contribution and
approximating the remaining part, we add a small tensor-power mixture
to restore trace one.
 
The first two approximation steps are
Lemmas~\ref{lem:remove-low-spectrum} and
\ref{lem:approximate-high-spectrum}.  We state them next and prove the
approximation theorem from them.  Their proofs are deferred to
Appendix~\ref{app:tensor-power-proofs}: the former rests on an
auxiliary bound for projected blocks of \(\rho_\ell\)
(Lemma~\ref{lem:projected-pattern-bound}), proved there by a swap
calculation, and the latter is derived from the de Finetti theorem~\cite{jeronimo2026optimal}
(Theorem~\ref{thm:jwx-tensor-power-approximation}).
 
\begin{restatable}[Removing the low-eigenvalue part]{lemma}{removelowspectrum}
\label{lem:remove-low-spectrum}
Let \(\Gamma_m\) be a density operator supported on
\(\Sym^m(\mathcal H)\).  Fix \(1\le\ell\le m\), and let
\(\rho_1\) and \(\rho_\ell\) be its reduced states on one and
\(\ell\) registers, respectively.  Let
\[
  P=\one_{(1/m,1]}(\rho_1).
\]
Here, \(\one_{(1/m,1]}(\rho_1)\) denotes the orthogonal projector
onto the span of the eigenvectors of \(\rho_1\) whose eigenvalues
are larger than \(1/m\).
Then
\[
  \left\|
    \rho_\ell-P^{\otimes\ell}\rho_\ell P^{\otimes\ell}
  \right\|_2
  \le
  \sqrt{2(2^\ell-1)}
  \left(
    \frac{3+2\ell(\ell-1)}m
  \right)^{1/4}.
\]
\end{restatable}
 
\begin{restatable}[Approximating the high-eigenvalue part]{lemma}{approximatehighspectrum}
\label{lem:approximate-high-spectrum}
Let \(\Gamma_m\) be a density operator supported on
\(\Sym^m(\mathcal H)\).  Fix \(1\le\ell\le m\), and let
\(\rho_\ell\) be its reduced state on \(\ell\) registers.
Let \(P\) be a nonzero orthogonal projector on \(\mathcal H\), and write
\[
  \rho_P=P^{\otimes\ell}\rho_\ell P^{\otimes\ell}.
\]
There is a subnormalized mixture \(\widetilde\rho_P\) of states
\((\proj\psi)^{\otimes\ell}\), with \(\ket\psi\in P\mathcal H\),
such that
\[
  \Tr\widetilde\rho_P=\Tr\rho_P,
  \qquad
  \lVert\rho_P-\widetilde\rho_P\rVert_1
  \le
  \frac{8\ell\sqrt{\operatorname{rank}P-1}}m.
\]
\end{restatable}
 
We now prove the approximation theorem from these two statements.
 
\begin{proof}[Proof of Theorem~\ref{lem:tensor-power-approximation}]
Let \(\rho_1\) be the one-register reduced state of \(\Gamma_m\), and
define
\[
  P=\one_{(1/m,1]}(\rho_1),
  \qquad
  Q=I-P,
  \qquad
  \rho_P=P^{\otimes\ell}\rho_\ell P^{\otimes\ell}.
\]
Here, \(P\) is the orthogonal projector onto the span of the
eigenvectors of \(\rho_1\) whose eigenvalues are larger than \(1/m\).
Since \(\Tr\rho_1=1\), the rank of \(P\) is at most \(m\).
 
Lemma~\ref{lem:remove-low-spectrum} gives
\[
  \left\|\rho_\ell-\rho_P\right\|_2
  \le
  \sqrt{2(2^\ell-1)}
  \left(
    \frac{3+2\ell(\ell-1)}m
  \right)^{1/4}.
\]
If \(P\ne0\), let \(\widetilde\rho_P\) denote the subnormalized
tensor-power mixture supplied by
Lemma~\ref{lem:approximate-high-spectrum}.  If \(P=0\), set
\(\widetilde\rho_P=0\).  In either case,
\[
  \Tr\widetilde\rho_P=\Tr\rho_P,
  \qquad
  \left\|\rho_P-\widetilde\rho_P\right\|_1
  \le\frac{8\ell}{\sqrt m}.
\]
 
The operator \(\rho_P\) contains only the part of \(\rho_\ell\) in
which all \(\ell\) registers lie in \(P\mathcal H\).  Consequently,
\(\widetilde\rho_P\) need not have trace one.  We now add a
tensor-power mixture of trace \(1-\Tr\rho_P\) to obtain a normalized
state.
 
The operator inequality
\[
  I-P^{\otimes\ell}
  \preceq
  \sum_{j=1}^{\ell}
  I^{\otimes(j-1)}\otimes Q\otimes
  I^{\otimes(\ell-j)}
\]
implies
\[
\begin{aligned}
  1-\Tr\rho_P
  &=
  \Tr\!\left[
    (I-P^{\otimes\ell})\rho_\ell
  \right] \\
  &\le
  \sum_{j=1}^{\ell}
  \Tr\!\left[
    \left(
      I^{\otimes(j-1)}\otimes Q\otimes
      I^{\otimes(\ell-j)}
    \right)\rho_\ell
  \right] \\
  &=
  \ell\,\Tr(Q\rho_1).
\end{aligned}
\]
The last equality follows because every one-register reduced state of
\(\rho_\ell\) is \(\rho_1\).
 
Diagonalize the part of \(\rho_1\) supported on \(Q\mathcal H\) as
\[
  Q\rho_1Q=\sum_i\lambda_i\proj{e_i}.
\]
By the definition of \(Q\), every \(\lambda_i\) is at most \(1/m\).
If \(\Tr(Q\rho_1)>0\), define
\[
  \sigma_\ell
  =
  \widetilde\rho_P
  +
  \frac{1-\Tr\rho_P}{\Tr(Q\rho_1)}
  \sum_i
  \lambda_i
  \bigl(\proj{e_i}\bigr)^{\otimes\ell}.
\]
The second term is a mixture of identical tensor powers whose trace is
exactly \(1-\Tr\rho_P\).  Hence \(\sigma_\ell\) has trace one and is a
mixture of identical tensor powers.  Since the vectors
\(\ket{e_i}^{\otimes\ell}\) are orthonormal,
\[
\begin{aligned}
  \left\|\sigma_\ell-\widetilde\rho_P\right\|_2^2
  &=
  \left(
    \frac{1-\Tr\rho_P}{\Tr(Q\rho_1)}
  \right)^2
  \sum_i\lambda_i^2 \\
  &\le
  \frac{(1-\Tr\rho_P)^2}
       {m\,\Tr(Q\rho_1)} \\
  &\le
  \frac{\ell(1-\Tr\rho_P)}m
  \le
  \frac\ell m.
\end{aligned}
\]
Here the first inequality uses
\[
  \sum_i\lambda_i^2
  \le
  \frac1m\sum_i\lambda_i
  =
  \frac{\Tr(Q\rho_1)}m,
\]
and the second uses
\(1-\Tr\rho_P\le\ell\Tr(Q\rho_1)\).
 
If \(\Tr(Q\rho_1)=0\), \(\Tr\rho_P=1\). In this case no additional term is needed, and we set \(\sigma_\ell=\widetilde\rho_P\).
 
Finally,
\[
\begin{aligned}
  \left\|\rho_\ell-\sigma_\ell\right\|_2
  &\le
  \left\|\rho_\ell-\rho_P\right\|_2
  +
  \left\|\rho_P-\widetilde\rho_P\right\|_2
  +
  \left\|\widetilde\rho_P-\sigma_\ell\right\|_2 \\
  &\le
  \sqrt{2(2^\ell-1)}
  \left(
    \frac{3+2\ell(\ell-1)}m
  \right)^{1/4}
  +
  \frac{8\ell+\sqrt\ell}{\sqrt m},
\end{aligned}
\]
where we used
\(\lVert X\rVert_2\le\lVert X\rVert_1\) for the middle term.
 
If \(\ell\le(1/2-\eta)\log_2m\), then
\(2^\ell\le m^{1/2-\eta}\), and the first term above is
\[
  O\!\left(
    \sqrt{\ell+1}\,m^{-\eta/2}
  \right).
\]
The second term is also at most
\(\operatorname{poly}(\log m)m^{-\eta/2}\), which proves the stated
rate.
\end{proof}

\section{Discussion and open problems}
\label{sec:discussion}

The first open problem is whether the soundness can be made exactly
\(1/4\), eliminating the inverse-polynomial factor.  A stronger and complexity-theoretically consequential question is whether one can cross the line \(c=4s\) furnished by sign removal by an inverse-polynomial amount, which would imply
\(\QMAreal(2)=\NEXP\).

A second question concerns one-proof amplification. Our proof uses the class equality \(\NEXP=\QMAplus(1,c_0,s_0)\)~\cite{bassirian2024itcs} for some small constant gap $c_0-s_0$, but it also crucially relies on unentanglement.
Is it possible to obtain mild gap amplification for $\QMAplus$? Note that strong gap amplification might be impossible, since it would imply $\mathsf{QMA}=\NEXP$. It would also be interesting to explore or improve gap amplification of other related classes such as $\mathsf{QMA_{IS}}$~\cite{bassirian2024quantum} or $\mathsf{pureQPH}$~\cite{grewal2026pure}.

Finally, Theorem~\ref{lem:tensor-power-approximation} may be
useful whenever a protocol depends only on second moments of reduced
states on logarithmically many registers.  Characterizing the largest
class of measurements for which such a bound replaces a
dimension-dependent trace-norm approximation for reduced states is an independent
question.

\bibliographystyle{alpha}
\bibliography{refs}

\appendix
 
\setcounter{theorem}{0}
\renewcommand{\thetheorem}{\thesection.\arabic{theorem}}
\makeatletter
\@ifundefined{thelemma}{}{\renewcommand{\thelemma}{\thesection.\arabic{lemma}}}
\makeatother
\setcounter{equation}{0}
\renewcommand{\theequation}{\thesection.\arabic{equation}}
 
\section{Deferred proofs from Section~\ref{sec:tensor-power-approximation}}
\label{app:tensor-power-proofs}
 
This appendix proves the two approximation steps used in the proof of
Theorem~\ref{lem:tensor-power-approximation}:
Lemma~\ref{lem:remove-low-spectrum} (removing the low-eigenvalue part)
and Lemma~\ref{lem:approximate-high-spectrum} (approximating the
high-eigenvalue part).
Section~\ref{app:proof-remove-low} proves Lemma~\ref{lem:remove-low-spectrum}.
Section~\ref{app:proof-approximate-high} proves
Lemma~\ref{lem:approximate-high-spectrum} from a theorem of Jeronimo,
Wu, and Xu~\cite{jeronimo2026optimal}.

\subsection{Two frequently used facts}
\label{app:toolbox}
 
The following two facts are used several times below.  First, for
every positive semidefinite \(\rho\) and orthogonal projectors
\(\Pi,\Pi'\),
\begin{equation}
\label{eq:hs-blocks}
  \|\Pi\rho\Pi'\|_2^2
  =\Tr(\rho\Pi\rho\Pi')
  =\Tr\!\left[
    (\rho^{1/2}\Pi\rho^{1/2})(\rho^{1/2}\Pi'\rho^{1/2})
  \right],
  \qquad
  \|\rho^{1/2}\Pi\rho^{1/2}\|_2
  =\|\Pi\rho\Pi\|_2.
\end{equation}
 
As before, \(X_i^{(1)}\)
(resp.\ \(X_j^{(2)}\)) denotes \(X\) acting on register \(i\) of the
first copy (resp.\ register \(j\) of the second copy), and
\(F_{ij}\) swaps register \(i\) of the first copy with register
\(j\) of the second copy. Second fact is the following.
 
\begin{lemma}
\label{lem:reduction}
Let \(\Gamma_m\) be a density operator supported on
\(\Sym^m(\mathcal H)\), and let \(\rho_\ell\) be its reduced state
on \(\ell\) registers, where \(1\le\ell\le m\).  Let
\(B_1,\ldots,B_\ell\) be operators on \(\mathcal H\) and set
\(B=B_1\otimes\cdots\otimes B_\ell\).  If
\(\mathbf i=(i_1,\ldots,i_\ell)\) and
\(\mathbf j=(j_1,\ldots,j_\ell)\) are lists of pairwise distinct
positions in \([m]\), then
\begin{equation}
\label{eq:reduction}
  \Tr\!\Bigl[
    (\Gamma_m\otimes\Gamma_m)
    \prod_{r=1}^{\ell}(B_r)_{i_r}^{(1)}(B_r)_{j_r}^{(2)}F_{i_rj_r}
  \Bigr]
  =\Tr(\rho_\ell B\rho_\ell B).
\end{equation}
\end{lemma}
 
\begin{proof}
Since \(\Gamma_m\) is supported on \(\Sym^m(\mathcal H)\), \(U_\pi\Gamma_mU_\pi^\dagger=\Gamma_m\), and conjugation relabels
registers: \(U_\pi X_i^{(1)}U_\pi^\dagger=X_{\pi(i)}^{(1)}\) and
\((U_\pi\otimes U_\sigma)F_{ij}(U_\pi\otimes U_\sigma)^\dagger
=F_{\pi(i)\sigma(j)}\).  Because the entries of each list are
pairwise distinct, we may choose \(\pi,\sigma\in S_m\) with
\(\pi(i_r)=r\) and \(\sigma(j_r)=r\) for all \(r\).  Moreover, on
product vectors, \(\prod_rF_{rr}\) exchanges the two blocks of
registers \(1,\ldots,\ell\), so
\(\prod_rF_{rr}=F^{\otimes\ell}\otimes I\).
Hence
\begin{align*}
  &\Tr\!\Bigl[
    (\Gamma_m\otimes\Gamma_m)
    \prod_{r=1}^{\ell}(B_r)_{i_r}^{(1)}(B_r)_{j_r}^{(2)}F_{i_rj_r}
  \Bigr]
  \\
  &\quad=
  \Tr\!\Bigl[
    (\Gamma_m\otimes\Gamma_m)
    (B\otimes B)\textstyle\prod_rF_{rr}
  \Bigr]
  \\
  &\quad=
  \Tr\!\left[
    (\rho_\ell\otimes\rho_\ell)(B\otimes B)F^{\otimes\ell}
  \right]
  =\Tr(\rho_\ell B\rho_\ell B).
\end{align*}
\end{proof}
 
\subsection{Proof of Lemma~\ref{lem:remove-low-spectrum}}
\label{app:proof-remove-low}
 
We restate the lemma.
 
\removelowspectrum*
 
We first give technical lemmas used for the proof.

\begin{lemma}
\label{lem:averaged-swaps}
Let \(P\) be an orthogonal projector on \(\mathcal H\), and set
\(Q=I-P\).  On
two copies of \(\Sym^m(\mathcal H)\), define
\[
  \mathcal T_P
  =
  \frac1{m^2}\sum_{i,j=1}^m
  P_i^{(1)}P_j^{(2)}F_{ij},
  \qquad
  \mathcal T_Q
  =
  \frac1{m^2}\sum_{i,j=1}^m
  Q_i^{(1)}Q_j^{(2)}F_{ij},
\]
where \(F_{ij}\) swaps register \(i\) of the first copy with register
\(j\) of the second copy.  On
\(\Sym^m(\mathcal H)\otimes\Sym^m(\mathcal H)\), these operators
commute and satisfy
\begin{equation}
\label{eq:averaged-swap-spectrum}
  -\frac1mI\preceq\mathcal T_P\preceq I,
  \qquad
  -\frac1mI\preceq\mathcal T_Q\preceq I.
\end{equation}
\end{lemma}

\begin{proof}[Proof of Lemma~\ref{lem:averaged-swaps}]
Throughout, we use the relabeling relations
\(F_{ij}P_i^{(1)}=P_j^{(2)}F_{ij}\) and
\(F_{ij}P_j^{(2)}=P_i^{(1)}F_{ij}\), and likewise with \(Q\) or with other positions.
 
We first prove the inequalities~\ref{eq:averaged-swap-spectrum}.
Conjugation by \(U_\pi\otimes U_\sigma\) maps the summand of
\(\mathcal T_P\) with index \((i,j)\) to the summand with index
\((\pi(i),\sigma(j))\), so \(\mathcal T_P\) commutes with every
\(U_\pi\otimes U_\sigma\) and hence preserves
\(\Sym^m(\mathcal H)\otimes\Sym^m(\mathcal H)\).  Each summand is
self-adjoint,
\[
  \bigl(P_i^{(1)}P_j^{(2)}F_{ij}\bigr)^\dagger
  =F_{ij}P_j^{(2)}P_i^{(1)}
  =P_i^{(1)}P_j^{(2)}F_{ij},
\]
so \(\mathcal T_P\) is self-adjoint. 
The same holds for \(\mathcal T_Q\).

Let \(\ket{\Psi}\in\Sym^m(\mathcal H)\otimes\Sym^m(\mathcal H)\) be a unit
vector.  Since \((U_\pi\otimes U_\sigma)\ket{\Psi}=\ket{\Psi}\), all \(m^2\)
summands of \(\mathcal T_P\) have the same expectation in \(\ket{\Psi}\),
so, setting \(\ket{\chi}=P_1^{(2)}\ket{\Psi}\),
\begin{align*}
  \langle\Psi | \mathcal T_P | \Psi\rangle
  =
  \bigl\langle\Psi | P_1^{(1)}P_1^{(2)}F_{11} | \Psi\bigr\rangle
  =
  \langle\chi | F_{11} | \chi\rangle
  \le 1,
\end{align*}
which proves \(\mathcal T_P\preceq I\) on
\(\Sym^m(\mathcal H)\otimes\Sym^m(\mathcal H)\).
 
For the lower bound, note that \(\ket{\chi}\) is invariant under
permutations of the \(m\) registers of the first copy. 
Conjugating \(F_{11}\) by these permutations therefore gives
\(\langle\chi | F_{11} | \chi\rangle=\langle\chi | F_{i1} | \chi\rangle:=t\) for
every \(i\in[m]\), so
\[
  \langle\Psi|\mathcal T_P | \Psi\rangle
  =\frac1m\sum_{i=1}^m\langle\chi | F_{i1} | \chi\rangle=t .
\]

Our goal is now to show \(t\ge-1/m\). 
Consider the \(m+1\) vectors
\[
  \ket{\chi},\;F_{11}\ket{\chi},\;F_{21}\ket{\chi},\;\ldots,\;F_{m1}\ket{\chi} .
\]
Each has norm \(\|\ket{\chi}\|\), and any two of them have inner product \(t\). Then,
\begin{align*}
  &0\le
  \Bigl\|\ket{\chi}+\sum_{i=1}^mF_{i1}\ket{\chi}\Bigr\|^2
  =(m+1)\|\chi\|^2+m(m+1)\,t.
\end{align*}
Hence
\(t\ge-\|\chi\|^2/m\ge-1/m\), that is,
\(\langle\Psi,\mathcal T_P\Psi\rangle\ge-1/m\). Together with
the upper bound, this proves
\[
-\frac1mI\preceq\mathcal T_P\preceq I
\] 
on \(\Sym^m(\mathcal H)\otimes\Sym^m(\mathcal H)\). The same argument
with \(Q\) in place of \(P\) handles \(\mathcal T_Q\).
 
Finally, we prove the commutativity. We compare
\(m^4\,\mathcal T_P\mathcal T_Q
=\sum_{i,j,k,l}
P_i^{(1)}P_j^{(2)}F_{ij}\,Q_k^{(1)}Q_l^{(2)}F_{kl}\)
with the reversed products \(m^4\,\mathcal T_Q\mathcal T_P\), term by term.  If \(k=i\), moving
\(Q_i^{(1)}\) through \(F_{ij}\) turns it into \(Q_j^{(2)}\), which
meets \(P_j^{(2)}\):
\[
  P_i^{(1)}P_j^{(2)}F_{ij}\,Q_i^{(1)}Q_l^{(2)}F_{il}
  =P_i^{(1)}\bigl(P_j^{(2)}Q_j^{(2)}\bigr)F_{ij}\,Q_l^{(2)}F_{il}
  =0,
\]
and in the reversed product
\(Q_i^{(1)}Q_l^{(2)}F_{il}\,P_i^{(1)}P_j^{(2)}F_{ij}\), moving
\(P_i^{(1)}\) through \(F_{il}\) produces
\(Q_l^{(2)}P_l^{(2)}=0\). The case \(l=j\) vanishes in the same way. 
If \(k\ne i\) and \(l\ne j\), the three factors
\(P_i^{(1)},P_j^{(2)},F_{ij}\) act on registers disjoint from those
of \(Q_k^{(1)},Q_l^{(2)},F_{kl}\), so the two groups commute and the
two products agree.  Summing over all indices,
\(\mathcal T_P\mathcal T_Q=\mathcal T_Q\mathcal T_P\) on \(\Sym^m(\mathcal H)\otimes\Sym^m(\mathcal H)\).
\end{proof}

Using~\Cref{lem:averaged-swaps}, we can show the following.
 
\begin{lemma}
\label{lem:projected-pattern-bound}
Let \(\Gamma_m\) be a density operator supported on
\(\Sym^m(\mathcal H)\), and let \(\rho_\ell\) and \(\rho_1\) be its
reduced states on \(\ell\) registers and one register, respectively,
where \(1\le\ell\le m\).  Let \(P\) be an orthogonal projector and set
\(Q=I-P\).  For
\(\mathbf S=(S_1,\ldots,S_\ell)\in\{P,Q\}^{\ell}\), write
\[
  \Pi_{\mathbf S}=S_1\otimes\cdots\otimes S_\ell.
\]
If \(\mathbf S\ne(P,\ldots,P)\), then
\begin{equation}
\label{eq:projected-pattern-bound}
  \left\|\Pi_{\mathbf S}\rho_\ell\Pi_{\mathbf S}\right\|_2^2
  \le
  \Tr[(Q\rho_1Q)^2]
  +\frac{2+2\ell(\ell-1)}m.
\end{equation}
\end{lemma}

\begin{proof}[Proof of Lemma~\ref{lem:projected-pattern-bound}]
Fix the pattern \(\mathbf S=(S_1,\ldots,S_\ell)\ne(P,\ldots,P)\),
let \(q\ge1\) be the number of indices \(r\) with \(S_r=Q\), and for
lists \(\mathbf i,\mathbf j\) of positions in \([m]\) write
\[
  W(\mathbf i,\mathbf j)
  =\prod_{r=1}^{\ell}
   (S_r)_{i_r}^{(1)}(S_r)_{j_r}^{(2)}F_{i_rj_r}.
\]
Each factor is a product of two projectors and a unitary, so
\(\|W(\mathbf i,\mathbf j)\|_\infty\le1\) by submultiplicativity.
 
If the entries of \(\mathbf i\) are pairwise distinct and likewise
for \(\mathbf j\), then Lemma~\ref{lem:reduction} with \(B_r=S_r\)
and \eqref{eq:hs-blocks} give
\[\Tr[(\Gamma_m\otimes\Gamma_m)W(\mathbf i,\mathbf j)]
=\Tr(\rho_\ell\Pi_{\mathbf S}\rho_\ell\Pi_{\mathbf S})
=\|\Pi_{\mathbf S}\rho_\ell\Pi_{\mathbf S}\|_2^2.
\]
Let \(M_{\mathrm{dist}}\) be the average of
\(W(\mathbf i,\mathbf j)\) over two independent uniformly random
ordered lists with pairwise distinct entries. Then,
\begin{equation}
\label{eq:mdist-value}
  \Tr\!\left[(\Gamma_m\otimes\Gamma_m)M_{\mathrm{dist}}\right]
  =\left\|\Pi_{\mathbf S}\rho_\ell\Pi_{\mathbf S}\right\|_2^2 .
\end{equation}
 
Let \(M_{\mathrm{ind}}\) be the corresponding average when all
\(2\ell\) entries are independent and uniform on \([m]\), with
repetitions allowed.  Then
\begin{align}
  M_{\mathrm{ind}}
  &=
  \frac1{m^{2\ell}}
  \sum_{i_1,j_1,\ldots,i_\ell,j_\ell=1}^m\,
  \prod_{r=1}^{\ell}
    (S_r)_{i_r}^{(1)}(S_r)_{j_r}^{(2)}F_{i_rj_r}
  \notag\\
  &=
  \prod_{r=1}^{\ell}\Bigl(
    \frac1{m^2}\sum_{i,j=1}^m
    (S_r)_{i}^{(1)}(S_r)_{j}^{(2)}F_{ij}
  \Bigr)
  &&\text{}
  \notag\\
  &=
  \prod_{r=1}^{\ell}\mathcal T_{S_r}
  =
  \mathcal T_P^{\ell-q}\mathcal T_Q^{q}.
  &&\text{(By Lemma~\ref{lem:averaged-swaps})}
  \label{eq:mind-value}
\end{align}

Next, for all \(s,t\in[-1/m,1]\) and \(q\ge1\),
\begin{equation}
\label{eq:scalar-ineq}
  s^{\ell-q}t^{q}\le t+\frac2m.
\end{equation}
By Lemma~\ref{lem:averaged-swaps}, \(\mathcal T_P\) and
\(\mathcal T_Q\) are commuting self-adjoint operators with spectra
in \([-1/m,1]\), so they have a joint orthonormal eigenbasis.
Applying \eqref{eq:scalar-ineq} to each pair of joint eigenvalues
gives
\begin{equation}
\label{eq:independent-pattern-bound}
  \mathcal T_P^{\ell-q}\mathcal T_Q^{q}
  \preceq
  \mathcal T_Q+\frac{2}{m}I .
\end{equation}
 
Call a pair of lists \emph{colliding} if some list repeats an entry.
Under independent sampling, the collision probability \(p\) is at
most \(2\binom\ell2/m=\ell(\ell-1)/m\) (union bound over the pairs
\(r<s\) in each list, \(\Pr[i_r=i_s]=1/m\)).  Conditioned on no
collision, the pair is uniform over pairs of lists with pairwise
distinct entries, so the law of total expectation gives
\(M_{\mathrm{ind}}=(1-p)M_{\mathrm{dist}}+p\,M_{\mathrm{coll}}\),
with \(M_{\mathrm{coll}}\) the collision-conditioned average.  Both
conditional averages have operator norm at most \(1\), being
averages of the \(W(\mathbf i,\mathbf j)\), so the triangle
inequality gives
\begin{equation}
\label{eq:average-comparison}
  \lVert M_{\mathrm{ind}}-M_{\mathrm{dist}}\rVert_\infty
  =p\,\lVert M_{\mathrm{coll}}-M_{\mathrm{dist}}\rVert_\infty
  \le2p
  \le\frac{2\ell(\ell-1)}m .
\end{equation}
Since \(|\Tr[(\Gamma_m\otimes\Gamma_m)D]|\le\|D\|_\infty\) for every
operator \(D\), and \(A\preceq B\) implies
\(\Tr[(\Gamma_m\otimes\Gamma_m)A]\le\Tr[(\Gamma_m\otimes\Gamma_m)B]\),
\begin{align}
  &\left\|\Pi_{\mathbf S}\rho_\ell\Pi_{\mathbf S}\right\|_2^2
  \notag\\
  &\quad=
  \Tr\!\left[(\Gamma_m\otimes\Gamma_m)M_{\mathrm{dist}}\right]
  \le
  \Tr\!\left[(\Gamma_m\otimes\Gamma_m)M_{\mathrm{ind}}\right]
  +\frac{2\ell(\ell-1)}m
  &&\text{(\eqref{eq:mdist-value} and
          \eqref{eq:average-comparison})}
  \notag\\
  &\quad\le
  \Tr\!\left[(\Gamma_m\otimes\Gamma_m)\mathcal T_Q\right]
  +\frac{2+2\ell(\ell-1)}m.
  &&\text{(\eqref{eq:mind-value} and
          \eqref{eq:independent-pattern-bound})}
  \label{eq:step3-bound}
\end{align}
Finally, a one-entry list has no repetitions, so
Lemma~\ref{lem:reduction} with \(\ell=1\), \(B_1=Q\) and \eqref{eq:hs-blocks} give
\[
\Tr[(\Gamma_m\otimes\Gamma_m)Q_i^{(1)}Q_j^{(2)}F_{ij}]
=\Tr(\rho_1Q\rho_1Q)=\Tr[(Q\rho_1Q)^2].
\] 
Averaging over \(i,j\),
\(\Tr[(\Gamma_m\otimes\Gamma_m)\mathcal T_Q]=\Tr[(Q\rho_1Q)^2]\),
and substituting this into \eqref{eq:step3-bound} proves the statement of the lemma, i.e., 
\eqref{eq:projected-pattern-bound}.
\end{proof}

We are now ready to prove Lemma~\ref{lem:remove-low-spectrum}.
 
\begin{proof}[Proof of Lemma~\ref{lem:remove-low-spectrum}]
Set \(Q=I-P\).  For
\(\mathbf S=(S_1,\ldots,S_\ell)\in\{P,Q\}^{\ell}\), write
\(\Pi_{\mathbf S}=S_1\otimes\cdots\otimes S_\ell\).
For every \(\mathbf S\ne(P,\ldots,P)\),
\begin{align}
  \|\Pi_{\mathbf S}\rho_\ell\Pi_{\mathbf S}\|_2^2
  \le
  \Tr[(Q\rho_1Q)^2]
  +\frac{2+2\ell(\ell-1)}m
  \le
  \frac{3+2\ell(\ell-1)}m,
  \label{eq:diagonal-block-bound}
\end{align}
where the first inequality is
Lemma~\ref{lem:projected-pattern-bound}, and the second holds because
\(\Tr[(Q\rho_1Q)^2]\) is at most
\(\|Q\rho_1Q\|_\infty\Tr(Q\rho_1Q)\le1/m\).
 
For any two patterns \(\mathbf S,\mathbf T\ne(P,\ldots,P)\),
positivity of \(\rho_\ell\) and Cauchy--Schwarz for the
Hilbert--Schmidt inner product give
\begin{align}
  \|\Pi_{\mathbf S}\rho_\ell\Pi_{\mathbf T}\|_2^2
  &=
  \Tr\!\left[
    (\rho_\ell^{1/2}\Pi_{\mathbf S}\rho_\ell^{1/2})
    (\rho_\ell^{1/2}\Pi_{\mathbf T}\rho_\ell^{1/2})
  \right]
  \le
  \|\rho_\ell^{1/2}\Pi_{\mathbf S}\rho_\ell^{1/2}\|_2\,
  \|\rho_\ell^{1/2}\Pi_{\mathbf T}\rho_\ell^{1/2}\|_2
  \notag\\
  &=
  \|\Pi_{\mathbf S}\rho_\ell\Pi_{\mathbf S}\|_2\,
  \|\Pi_{\mathbf T}\rho_\ell\Pi_{\mathbf T}\|_2
  \le
  \frac{3+2\ell(\ell-1)}m.
  \qquad
  \text{(\eqref{eq:hs-blocks} and
        \eqref{eq:diagonal-block-bound})}
  \label{eq:offdiagonal-block-bound}
\end{align}
The blocks are pairwise orthogonal in the Hilbert--Schmidt inner
product: for \((\mathbf S,\mathbf T)\ne(\mathbf U,\mathbf V)\),
\begin{equation}
\label{eq:block-orthogonality}
  \left\langle
    \Pi_{\mathbf S}\rho_\ell\Pi_{\mathbf T},
    \Pi_{\mathbf U}\rho_\ell\Pi_{\mathbf V}
  \right\rangle_{\mathrm{HS}}
  =
  \Tr\!\left[
    \Pi_{\mathbf T}\rho_\ell
    \Pi_{\mathbf S}\Pi_{\mathbf U}
    \rho_\ell\Pi_{\mathbf V}
  \right]
  =0.
\end{equation}
Since
\(I=(P+Q)^{\otimes\ell}\) gives
\(I-P^{\otimes\ell}=\sum_{\mathbf S\ne(P,\ldots,P)}\Pi_{\mathbf S}\),
\begin{align*}
  \left\|
    (I-P^{\otimes\ell})\rho_\ell(I-P^{\otimes\ell})
  \right\|_2^2
  &=
  \sum_{\substack{\mathbf S,\mathbf T\\
                   \mathbf S,\mathbf T\ne(P,\ldots,P)}}
  \|\Pi_{\mathbf S}\rho_\ell\Pi_{\mathbf T}\|_2^2
  &&\text{(\eqref{eq:block-orthogonality})}
  \\
  &\le
  (2^\ell-1)^2\frac{3+2\ell(\ell-1)}m.
  &&\text{(\eqref{eq:offdiagonal-block-bound})}
\end{align*}
Since \(\|A\|_\infty\le\|A\|_2\), this implies
\[
  \left\|
    (I-P^{\otimes\ell})\rho_\ell(I-P^{\otimes\ell})
  \right\|_\infty
  \le
  (2^\ell-1)
  \sqrt{\frac{3+2\ell(\ell-1)}m}.
\]
 
Fact~2.3 of \cite{jeronimo2026optimal} states that, for every density operator
\(\rho\) and orthogonal projector \(\Pi\),
\[
  \|\rho-\Pi\rho\Pi\|_2^2
  \le
  2\|(I-\Pi)\rho(I-\Pi)\|_\infty.
\]
Applying it with \(\rho=\rho_\ell\) and
\(\Pi=P^{\otimes\ell}\) proves
\[
  \left\|
    \rho_\ell-P^{\otimes\ell}\rho_\ell P^{\otimes\ell}
  \right\|_2^2
  \le
  2(2^\ell-1)
  \sqrt{\frac{3+2\ell(\ell-1)}m},
\]
as required.
\end{proof}

\subsection{Proof of Lemma~\ref{lem:approximate-high-spectrum}}
\label{app:proof-approximate-high}
 
We restate the lemma.
 
\approximatehighspectrum*
 
For the proof, we use the following result of Jeronimo, Wu, and Xu as
a black box.
 
\begin{theorem}[Jeronimo--Wu--Xu
  {\cite[Theorem~1.3]{jeronimo2026optimal}}]
\label{thm:jwx-tensor-power-approximation}
Let \(2\le \ell\le a\), and let \(\Omega_a\) be a density operator
supported on \(\Sym^a(\mathcal H)\), where
\(d=\dim\mathcal H\).  If \(\Omega_\ell\) is its reduced state on
\(\ell\) registers, then there is a probability measure \(\nu\) on
the unit sphere of \(\mathcal H\) such that
\[
  \left\|
    \Omega_\ell
    -
    \int
      \bigl(\proj{\psi}\bigr)^{\otimes\ell}
      \,d\nu(\psi)
  \right\|_1
  \le
  \frac{8\ell\sqrt{d-1}}{a}.
\]
\end{theorem}
 
\begin{proof}[Proof of Lemma~\ref{lem:approximate-high-spectrum}]
Write \(r=\operatorname{rank}P\) and \(Q=I-P\).  If \(r=0\), then
\(\rho_P=0\), and the conclusion holds with
\(\widetilde\rho_P=0\).  We therefore assume \(r\ge1\).

The approximation theorem used below applies to the
\(\ell\)-register reduction of a symmetric state supported entirely
on \(P\mathcal H\).  We cannot apply it directly to \(\Gamma_m\) in
the form needed here, because \(\Gamma_m\) need not be supported on
\((P\mathcal H)^{\otimes m}\).

We overcome this by separating \(\Gamma_m\) according to the total
number \(a\) of \(P\)-outcomes under the measurement \(\{P,Q\}\).
For each fixed \(a\), we obtain an \(a\)-register symmetric state
supported on \(P\mathcal H\), to which the approximation theorem can
be applied.  We then recover \(\rho_P\) from the corresponding
\(\ell\)-register reductions, approximate each reduction, and
recombine them with the same weights.

\medskip
\noindent
\emph{Constructing the \(a\)-sector states.}
For \(i\in[m]\), let \(P_i\) and \(Q_i\) denote \(P\) and \(Q\)
acting on register \(i\).  For \(S\subseteq[m]\), define
\[
  \Pi^{(S)}
  =
  \prod_{i\in S}P_i
  \prod_{i\notin S}Q_i.
\]
Thus, \(\Pi^{(S)}\) is the projector corresponding to obtaining
outcome \(P\) precisely on the registers in \(S\).

Write \(\Tr_{>k}\) for the partial trace over registers
\(k+1,\ldots,m\).  For \(0\le a\le m\), define
\begin{equation}
\label{eq:omega-def}
  \Omega_a
  =
  \binom ma
  \Tr_{>a}\!\left[
    \Pi^{([a])}\Gamma_m\Pi^{([a])}
  \right],
  \qquad
  p_a=\Tr\Omega_a.
\end{equation}
Here \([a]=\{1,\ldots,a\}\) is used as a representative choice of
the \(a\) registers with outcome \(P\), and the factor \(\binom ma\)
accounts for all possible choices of those registers.

We first check that \(p_a\) is the probability of obtaining exactly
\(a\) outcomes equal to \(P\).  Since \(\Pi^{([a])}\) is a
projector,
\[
  p_a
  =
  \binom ma
  \Tr\!\left[
    \Gamma_m\Pi^{([a])}
  \right].
\]
Permutation symmetry implies that
\[
  \Tr\!\left[
    \Gamma_m\Pi^{(S)}
  \right]
\]
has the same value for every \(S\subseteq[m]\) of size \(a\).
Therefore,
\[
  p_a
  =
  \sum_{\substack{S\subseteq[m]\\|S|=a}}
  \Tr\!\left[
    \Gamma_m\Pi^{(S)}
  \right].
\]
Moreover,
\[
  \sum_{S\subseteq[m]}\Pi^{(S)}
  =
  \prod_{i=1}^m(P_i+Q_i)
  =
  I.
\]
It follows that
\[
  p_a\ge0,
  \qquad
  \sum_{a=0}^m p_a=1.
\]

For \(p_a>0\), we intend to apply
Theorem~\ref{thm:jwx-tensor-power-approximation} to
\(\Omega_a/p_a\).  We must therefore verify that this is an
\(a\)-register symmetric state supported on \(P\mathcal H\).

The projectors \(P_i\) on the first \(a\) registers are already
present in \(\Pi^{([a])}\), so
\[
  \Omega_a
  =
  P^{\otimes a}\Omega_aP^{\otimes a}.
\]
Thus, \(\Omega_a\) is supported on
\((P\mathcal H)^{\otimes a}\).

Now let \(\pi\) be a permutation of the first \(a\) registers,
extended by the identity on the remaining registers.  The unitary
\(U_\pi\) commutes with \(\Pi^{([a])}\).  Since \(\Gamma_m\) is
supported on the symmetric subspace,
\[
  U_\pi\Gamma_m=\Gamma_m.
\]
Consequently,
\[
  U_\pi\Omega_a=\Omega_a.
\]
Averaging this equality over all permutations of the first \(a\)
registers gives
\[
  \Pi_{\mathrm{sym}}^{(a)}\Omega_a=\Omega_a.
\]
Taking adjoints also gives
\[
  \Omega_a\Pi_{\mathrm{sym}}^{(a)}=\Omega_a.
\]
Therefore, \(\Omega_a\) is supported on
\(\Sym^a(P\mathcal H)\).  In particular, whenever \(p_a>0\),
\(\Omega_a/p_a\) is a density operator on
\(\Sym^a(P\mathcal H)\).

\medskip
\noindent
\emph{Recovering \(\rho_P\) from the sector states.}
We next express \(\rho_P\) as a weighted sum of the
\(\ell\)-register reductions of the states constructed above.  By
the definition of \(\rho_P\),
\[
  \rho_P
  =
  P^{\otimes\ell}\rho_\ell P^{\otimes\ell}
  =
  \Tr_{>\ell}\!\left[
    \left(P^{\otimes\ell}\otimes I^{\otimes(m-\ell)}\right)
    \Gamma_m
    \left(P^{\otimes\ell}\otimes I^{\otimes(m-\ell)}\right)
  \right].
\]
Using \(I=P+Q\) on each of the last \(m-\ell\) registers gives
\begin{align*}
  P^{\otimes\ell}\otimes I^{\otimes(m-\ell)}
  &=
  \left(\prod_{i=1}^{\ell}P_i\right)
  \prod_{i=\ell+1}^m(P_i+Q_i)\\
  &=
  \sum_{T\subseteq\{\ell+1,\ldots,m\}}
  \Pi^{([\ell]\cup T)}.
\end{align*}
This expansion lists all possible \(P/Q\)-patterns on the last
\(m-\ell\) registers while keeping outcome \(P\) on each of the
first \(\ell\) registers.

Substituting the expansion on both sides of \(\Gamma_m\) and using
linearity of the partial trace gives
\begin{align*}
  \rho_P
  &=
  \Tr_{>\ell}\!\left[
    \left(
      \sum_{T\subseteq\{\ell+1,\ldots,m\}}
      \Pi^{([\ell]\cup T)}
    \right)
    \Gamma_m
    \left(
      \sum_{T'\subseteq\{\ell+1,\ldots,m\}}
      \Pi^{([\ell]\cup T')}
    \right)
  \right]\\
  &=
  \sum_{\substack{
    T\subseteq\{\ell+1,\ldots,m\}\\
    T'\subseteq\{\ell+1,\ldots,m\}
  }}
  \Tr_{>\ell}\!\left[
    \Pi^{([\ell]\cup T)}
    \Gamma_m
    \Pi^{([\ell]\cup T')}
  \right].
\end{align*}

The next step removes the cross terms between distinct
\(P/Q\)-patterns.  Suppose \(T\ne T'\).  Choose a register
\(k>\ell\) that belongs to exactly one of \(T\) and \(T'\).
Without loss of generality, suppose that \(k\in T\) and
\(k\notin T'\).  Then
\[
  P_k\Pi^{([\ell]\cup T)}
  =
  \Pi^{([\ell]\cup T)},
  \qquad
  \Pi^{([\ell]\cup T')}Q_k
  =
  \Pi^{([\ell]\cup T')}.
\]
Since register \(k\) is traced out, \(Q_k\) can be moved cyclically
inside the partial trace.  Hence
\begin{align*}
  &\Tr_{>\ell}\!\left[
    \Pi^{([\ell]\cup T)}
    \Gamma_m
    \Pi^{([\ell]\cup T')}
  \right]\\
  &=
  \Tr_{>\ell}\!\left[
    P_k
    \Pi^{([\ell]\cup T)}
    \Gamma_m
    \Pi^{([\ell]\cup T')}
    Q_k
  \right]\\
  &=
  \Tr_{>\ell}\!\left[
    Q_kP_k
    \Pi^{([\ell]\cup T)}
    \Gamma_m
    \Pi^{([\ell]\cup T')}
  \right]\\
  &=0,
\end{align*}
where the last equality uses \(Q_kP_k=0\).  Thus, only the terms
with \(T=T'\) remain:
\begin{equation}
\label{eq:rhoP-pattern-sum}
  \rho_P
  =
  \sum_{T\subseteq\{\ell+1,\ldots,m\}}
  \Tr_{>\ell}\!\left[
    \Pi^{([\ell]\cup T)}
    \Gamma_m
    \Pi^{([\ell]\cup T)}
  \right].
\end{equation}

We now group the surviving patterns according to their total number
of \(P\)-outcomes.  Fix \(a\ge\ell\).  The patterns with a total of
\(a\) outcomes equal to \(P\) are precisely those for which
\[
  |T|=a-\ell.
\]
Any two such sets \(T\) are related by a permutation of the last
\(m-\ell\) registers.  Such a permutation fixes the first \(\ell\)
registers, leaves \(\Gamma_m\) invariant, and does not change the
partial trace.  Hence all terms with \(|T|=a-\ell\) are equal.

There are
\[
  \binom{m-\ell}{a-\ell}
\]
sets \(T\) of this size.  Using
\(T=\{\ell+1,\ldots,a\}\) as a representative in each group,
equation~\eqref{eq:rhoP-pattern-sum} becomes
\[
  \rho_P
  =
  \sum_{a=\ell}^m
  \binom{m-\ell}{a-\ell}
  \Tr_{>\ell}\!\left[
    \Pi^{([a])}\Gamma_m\Pi^{([a])}
  \right].
\]

We next relate the operator in this sum to \(\Omega_a\). By
\eqref{eq:omega-def},
\begin{align*}
  (\Omega_a)_\ell
  &=
  \Tr_{\ell+1,\ldots,a}\!\left[
    \Omega_a
  \right]\\
  &=
  \binom ma
  \Tr_{\ell+1,\ldots,a}\!\left[
    \Tr_{a+1,\ldots,m}\!\left[
      \Pi^{([a])}\Gamma_m\Pi^{([a])}
    \right]
  \right]\\
  &=
  \binom ma
  \Tr_{\ell+1,\ldots,m}\!\left[
    \Pi^{([a])}\Gamma_m\Pi^{([a])}
  \right]\\
  &=
  \binom ma
  \Tr_{>\ell}\!\left[
    \Pi^{([a])}\Gamma_m\Pi^{([a])}
  \right],
\end{align*}
which gives
\begin{align}
  \rho_P
  &=
  \sum_{a=\ell}^m
  \frac{\binom{m-\ell}{a-\ell}}{\binom ma}
  (\Omega_a)_\ell
  \notag\\
  &=
  \sum_{a=\ell}^m
  \frac{\binom a\ell}{\binom m\ell}
  (\Omega_a)_\ell,
\label{eq:rhoP-decomposition}
\end{align}
where we used
\[
  \frac{\binom{m-\ell}{a-\ell}}{\binom ma}
  =
  \frac{\binom a\ell}{\binom m\ell}.
\]

We rewrite this decomposition in the form that will be used for the
approximation.  Since \(\Omega_a\ge0\), the condition \(p_a=0\)
implies \(\Omega_a=0\), so such terms vanish.  Therefore,
\[
  \rho_P
  =
  \sum_{\substack{a=\ell\\p_a>0}}^m
  p_a
  \frac{\binom a\ell}{\binom m\ell}
  \frac{(\Omega_a)_\ell}{p_a}.
\]
The three factors in each summand have the following roles.
The number \(p_a\) is the probability of obtaining exactly \(a\)
outcomes equal to \(P\), and
\((\Omega_a)_\ell/p_a\) is the corresponding normalized
\(\ell\)-register state.

The remaining factor specifies how much of the \(a\)-sector
contributes to \(\rho_P\).  Indeed, conditioned on there being
exactly \(a\) outcomes equal to \(P\), permutation symmetry makes
their locations uniform among the \(\binom ma\) subsets of size
\(a\).  Among these subsets,
\[
  \binom{m-\ell}{a-\ell}
\]
contain all of the first \(\ell\) positions.  Consequently,
\[
  \frac{\binom a\ell}{\binom m\ell}
  =
  \frac{\binom{m-\ell}{a-\ell}}{\binom ma}
\]
is the conditional probability that the first \(\ell\) outcomes
are all equal to \(P\).  This is relevant because the projection
defining \(\rho_P\) selects exactly these patterns.  Thus,
\[
  p_a\frac{\binom a\ell}{\binom m\ell}
\]
is precisely the weight with which
\((\Omega_a)_\ell/p_a\) contributes to \(\rho_P\).

\medskip
\noindent
\emph{Constructing the approximation.}
The last decomposition tells us exactly what to approximate.  For
each \(a\), we keep the scalar weight
\[
  p_a\frac{\binom a\ell}{\binom m\ell}
\]
unchanged and approximate only the normalized state
\((\Omega_a)_\ell/p_a\).

Fix \(a\ge\ell\) with \(p_a>0\).  We have already shown that
\(\Omega_a/p_a\) is a density operator supported on
\(\Sym^a(P\mathcal H)\), and
\((\Omega_a)_\ell/p_a\) is its reduction to \(\ell\) registers.

If \(r=1\), the space \(P\mathcal H\) is one-dimensional, so
\((\Omega_a)_\ell/p_a\) is already an exact tensor power.  If
\(\ell=1\), the spectral decomposition of
\((\Omega_a)_1/p_a\) expresses it exactly as a mixture of pure
states in \(P\mathcal H\).  In either case, we can choose a
probability measure \(\nu_a\) on the unit sphere of
\(P\mathcal H\) for which the approximation error below is zero.

In the remaining case \(r\ge2\) and \(\ell\ge2\),
Theorem~\ref{thm:jwx-tensor-power-approximation}, applied on the
\(r\)-dimensional space \(P\mathcal H\), supplies a probability
measure \(\nu_a\) on the unit sphere of \(P\mathcal H\) such that
\begin{equation}
\label{eq:sector-approximation}
  \left\|
    \frac{(\Omega_a)_\ell}{p_a}
    -
    \int
      \bigl(\proj{\psi}\bigr)^{\otimes\ell}
      \,d\nu_a(\psi)
  \right\|_1
  \le
  \frac{8\ell\sqrt{r-1}}{a}.
\end{equation}

We now replace each normalized state
\((\Omega_a)_\ell/p_a\) in the decomposition of \(\rho_P\) by its
tensor-power approximation and define
\[
  \widetilde\rho_P
  =
  \sum_{\substack{a=\ell\\p_a>0}}^m
  p_a\frac{\binom a\ell}{\binom m\ell}
  \int
    \bigl(\proj{\psi}\bigr)^{\otimes\ell}
    \,d\nu_a(\psi),
\]
which is a subnormalized mixture of states
\((\proj{\psi})^{\otimes\ell}\) with
\(\ket\psi\in P\mathcal H\).
Moreover, it follows from
\eqref{eq:rhoP-decomposition} that
\[
  \Tr\widetilde\rho_P
  =
  \sum_{a=\ell}^m
  p_a\frac{\binom a\ell}{\binom m\ell}
  =
  \Tr\rho_P.
\]

\medskip
\noindent
\emph{Bounding the approximation error.}
Because \(\rho_P\) and \(\widetilde\rho_P\) use the same sector
weights, their difference is the weighted sum of the errors made in
the individual sectors.  Therefore,
\begin{align*}
  \left\|\rho_P-\widetilde\rho_P\right\|_1
  &=
  \left\|
    \sum_{\substack{a=\ell\\p_a>0}}^m
    p_a\frac{\binom a\ell}{\binom m\ell}
    \left(
      \frac{(\Omega_a)_\ell}{p_a}
      -
      \int
        \bigl(\proj{\psi}\bigr)^{\otimes\ell}
        \,d\nu_a(\psi)
    \right)
  \right\|_1
  &&\text{(\eqref{eq:rhoP-decomposition} and the definition of
  \(\widetilde\rho_P\))}
  \\
  &\le
  \sum_{\substack{a=\ell\\p_a>0}}^m
  p_a\frac{\binom a\ell}{\binom m\ell}
  \left\|
    \frac{(\Omega_a)_\ell}{p_a}
    -
    \int
      \bigl(\proj{\psi}\bigr)^{\otimes\ell}
      \,d\nu_a(\psi)
  \right\|_1
  &&\text{(triangle inequality)}
  \\
  &\le
  8\ell\sqrt{r-1}
  \sum_{a=\ell}^m
  \frac{p_a}{a}
  \frac{\binom a\ell}{\binom m\ell}
  &&\text{(\eqref{eq:sector-approximation})}
  \\
  &=
  \frac{8\ell\sqrt{r-1}}{m}
  \sum_{a=\ell}^m
  p_a
  \frac{\binom{a-1}{\ell-1}}
       {\binom{m-1}{\ell-1}}
  &&\text{\(\left(
    \frac1a\frac{\binom a\ell}{\binom m\ell}
    =
    \frac1m
    \frac{\binom{a-1}{\ell-1}}
         {\binom{m-1}{\ell-1}}
  \right)\)}
  \\
  &\le
  \frac{8\ell\sqrt{r-1}}{m}
  &&\text{\(\left(
    \frac{\binom{a-1}{\ell-1}}
         {\binom{m-1}{\ell-1}}
    \le1,\quad
    \sum_{a=\ell}^m p_a
    \le\sum_{a=0}^m p_a=1
  \right)\)}.
\end{align*}

\end{proof}

\end{document}